\documentclass[11pt]{article}

\usepackage[utf8]{inputenc}

\usepackage[letterpaper, margin=1in]{geometry}

\usepackage{hyperref}
\usepackage{nameref}
\makeatletter
\let\orgdescriptionlabel\descriptionlabel
\renewcommand*{\descriptionlabel}[1]{
  \let\orglabel\label
  \let\label\@gobble
  \phantomsection
  \edef\@currentlabel{#1}
  \let\label\orglabel
  \orgdescriptionlabel{#1}
}
\makeatother

\usepackage{multirow}
\usepackage[utf8]{inputenc}
\usepackage[table,xcdraw]{xcolor}
\usepackage{amsthm,amsmath,amsfonts} 
\usepackage{bbm}
\usepackage{dsfont}
\usepackage{verbatim}
\usepackage{mdframed}

\providecommand{\anonalt}[2]{#1}

\newcommand{\solverLink}{\anonalt{\url{https://github.com/HamidrezaKmK/PhysarumSDPSolver}}{URL removed for double blindness}}

\newcommand{\dynamicOne}{first ansatz}
\newcommand{\dynamicTwo}{second ansatz}
\newcommand{\dynamicOneCap}{First Ansatz}
\newcommand{\dynamicTwoCap}{Second Ansatz}

\newcommand{\dynamicLP}{LP Physarum dynamic}
\newcommand{\dynamicSDP}{SDP Physarum dynamic}

\newcommand{\dynamicSDPs}{SDP Physarum dynamics}

\newcommand{\dynamicSDPCaps}{SDP Physarum Dynamics}

\providecommand{\vectorize}[1]{\underrightarrow{#1}}
\providecommand{\capitalize}[1]{\mathcal{#1}}
\newtheorem*{rep@theorem}{\rep@title}
\newcommand{\newreptheorem}[2]{
\newenvironment{rep#1}[1]{
 \def\rep@title{#2 \ref{##1}}
 \begin{rep@theorem}}
 {\end{rep@theorem}}}
\makeatother

\newtheorem{theorem}{Theorem}[section]
\newreptheorem{theorem}{Theorem}
\newtheorem{corollary}[theorem]{Corollary}
\newtheorem{lemma}[theorem]{Lemma}
\newreptheorem{lemma}{Lemma}

\newtheorem{conjecture}[theorem]{Conjecture}
\newreptheorem{conjecture}{Conjecture}
\newtheorem{remark}[theorem]{Remark}
\newreptheorem{remark}{Remark}

\def\norm#1{\left\| #1 \right\|}

\usepackage[ruled,vlined]{algorithm2e}

\usepackage{url}

\usepackage{enumerate}
\usepackage{mathtools}

\title{Physarum Inspired Dynamics to Solve Semi-Definite Programs}
\author{
    \anonalt{
    \anonalt{Yuan Gao\thanks{\noindent Max Planck Institute for Informatics}}{Anonymous author}
    \\
	\anonalt{\texttt{yuanagao@mpi-inf.mpg.de}}{anonymous email address}
	\and 
	\anonalt{Hamidreza Kamkari\thanks{\noindent Sharif University of Technology}}{Anonymous author } \\
	\anonalt{\texttt{hamidrezakamkari@gmail.com}}{anonymous email address}
	\and
	\anonalt{Andreas Karrenbauer\footnotemark[1]}{Anonymous author}\\
	\anonalt{\texttt{karrenba@mpi-inf.mpg.de}}{anonymous email address}
	\and
	\anonalt{Kurt Mehlhorn\footnotemark[1]}{Anonymous author}\\
	\anonalt{\texttt{mehlhorn@mpi-inf.mpg.de}}{anonymous email address}
	\and
	\anonalt{Mohammadamin Sharifi\footnotemark[2]}{Anonymous author}\\
	\anonalt{\texttt{sharifim689@gmail.com}}{anonymous email address}
	}{Anonymous Authors}
}
\date{}
\begin{document}

\maketitle
\thispagestyle{empty}
\begin{abstract}
	Physarum Polycephalum is a slime mold that can solve shortest path problems. 
	A mathematical model based on Physarum's behavior, known as the Physarum Directed Dynamics, 
	can solve positive linear programs. In this paper, we present a family of Physarum-based 
	dynamics extending the previous work and introduce a new algorithm to solve \emph{positive} 
	Semi-Definite Programs (SDP).  The Physarum dynamics are governed by orthogonal projections (w.r.t.\ time-dependent scalar products) on the affine subspace defined by the linear constraints.
	We present a natural generalization of the scalar products used in the LP case to the matrix space for SDPs, which boils down to the linear case when all matrices in the SDP are diagonal, thus, representing an LP.
	We investigate the behavior of the induced dynamics theoretically and experimentally, highlight challenges arising from the non-commutative nature of matrix products, and prove soundness and convergence under mild conditions.
	Moreover, we consider a more abstract view on the dynamics that suggests a slight variation to guarantee unconditional soundness and convergence-to-optimality.
	By simulating these dynamics using suitable discretizations, one obtains numerical algorithms for solving positive SDPs, which have applications in discrete optimization, e.g., for computing the Goemans-Williamson approximation for MaxCut or the Lov\'asz theta number for determining the clique/chromatic number in perfect graphs.
\end{abstract}
\clearpage

\pagenumbering{arabic}

\section{Introduction}

The Physarum computing model is an analog computing model motivated by the network
dynamics of the slime mold Physarum Polycephalum. In wet-lab experiments, it was observed that the slime mold is 
able to solve shortest path problems \cite{nakagaki2007minimum}. A mathematical model for the dynamic
behavior of the slime was proposed in \cite{tero2007mathematical}. 
Their slime network model is mathematically equivalent to an electrical network
with time-varying resistors that react to the amount of electrical current flowing through them.

A variant of the Physarum dynamics, the directed Physarum dynamics, is known to solve positive linear programs in
standard form \cite{johannson2012slime, straszak2015natural}. A positive linear program seeks to minimize a linear 
function $c^T x$ with a positive cost vector $c \in \mathbb{R}^n_{> 0}$
subject to the constraints $Ax = b$ and $x \ge 0$. Here $A \in \mathbb{R}^{m\times n} $and $b \in \mathbb{R}^{m}$. 

Since Semi-Definite Programming (SDP) problems generalize linear programs, it is natural to ask for what subclass of SDPs this dynamics can be adapted to. Before we turn to SDPs, we briefly review the dynamics for LPs.

The Physarum dynamics can be described by an autonomous dynamical system, i.e., its evolution over time is governed by a system of ordinary differential equations of the form
\[
\dot{x}(t) = v(x(t))
\]
where the \emph{velocities} $v(\cdot)$ do not explicitly depend on $t$.
For sake of of presentation, we may simply write $x$ instead of $x(t)$. Moreover, we first consider the special case when $Ax=b$ is satisfied at any time. 
The Physarum dynamics for LPs can then be written as~\cite{johannson2012slime, straszak2015natural}:
\begin{equation}
	\label{eq:physarum-dynamics-lp}
	\dot{x} = G(A^T (AGA^T)^\dagger Ax - c),
\end{equation}
where $G=C^{-1} X$ is the \textbf{Conductance} while $C$ and $X$ are the diagonal matrices
defined by the objective coefficient $c$ and the current point $x$, respectively. Equivalently, we also have:
\begin{equation}
	\label{eq:physarum-dynamics-lp-proj}
	\dot{x}(t)=-(I-GA^T (AGA^T)^\dagger A)w,
\end{equation}
where $w\coloneqq Gc$. Note that $(I-GA^T (AGA^T)^\dagger A)$ is an orthogonal projection 
on the kernel of $A$ under a weighted inner product.
It has been shown in~\cite{johannson2012slime, straszak2015natural} that the 
trajectory of the Physarum dynamic~\eqref{eq:physarum-dynamics-lp-proj} stays within the feasible 
region and converges to an optimal solution of the linear program. Numerous results are driven by 
altering the dynamics or taking a more abstract view~\cite{bonifaci2022physarum,bonifaci2017revised, 
KARRENBAUER2020260, StraszakV22}. 
These results motivate the research into generalizing 
the \dynamicLP{} to other more powerful convex optimization problems. In this paper, we 
define and analyze a family of Physarum inspired dynamics for positive SDP. 

\paragraph{Positive SDP} In general, an SDP can be written as one of the primal-dual pair
\begin{equation}\label{eq:SDP-definition}
\begin{split}
& \min \{ tr( C^T X ) : tr( A_{\ell}^T X ) = b_\ell \forall \ell \in [m],\, X \succeq 0 \} \\
\ge &\max \{ \sum_{\ell = 1}^m b_\ell y_\ell : \sum_{\ell = 1}^m y_\ell A_{\ell} + S = C, \, S \succeq 0 \},
\end{split}
\end{equation} 
where we may assume without loss of generality that $C$ and all $A_{\ell}$ are symmetric. 
This is due to the fact that $X$ is symmetric and each formula of form $tr(MX)$ is equivalent to 
$tr(\frac{M + M^T}{2} X)$; furthermore, we rewrite all of these conditions in this form 
where $C$ will be replaced by $\frac{C + C^T}{2}$ and each $A_{\ell}$ matrices will 
be replaced by $\frac{A_\ell + A_\ell^T}{2}$. We will call an SDP \textbf{positive} 
if the cost matrix $C$ is positive definite. 
Note that in this case, strong duality always holds because Slater's condition is trivially satisfied for the dual (maximization) problem as $y=0$, $S=C \succ 0$ is a solution in the relative interior.
Moreover, if $C$ and all the matrices $A_\ell$ are diagonal, 
then it suffices to only consider diagonal solutions $X$ and the SDP reduces to an LP. In this case, $X \succeq 0$ is equivalent to each of the diagonal elements 
being non-negative; which is the same as $diag(X) \ge 0$.
The matrix trace function then reduces to a dot product of the diagonal vectors.
Hence, the Physarum dynamics extends to positive SDPs restricted to diagonal matrices. But does it also extend to a more substantial subclass? Since the generalization of the conductances leaves some degrees of freedom, we will present a \dynamicOne{} and the subclass of positive SDPs for which we are able to prove soundness and convergence. Moreover, we report on computational experiments that suggest that this ansatz could actually work for all positive SDPs. However, to get there we make a \dynamicTwo{} that allows us to prove soundness and convergence for all positive SDPs.

Positive SDP is an important subclass of SDP, e.g., they are general enough for the Goemans-Williamson algorithm 
for approximating MaxCut~\cite{GW1995} or for computing the Lov\'asz theta number to determine the clique/chromatic number in a perfect graph~\cite{lovasz1979shannon,kleinberg1998lovasz}. This might not be obvious at a first glance but there are equivalent formulations with positive definite $C$ in both cases because a suitable multiple of the identity can be added to the cost matrix, which only causes a shift in the objective value since the diagonal of $X$ is fixed to all ones. In fact, whenever the dual SDP has a non-empty relative interior and one can find a witness with reasonable effort, one can replace the cost matrix $C$ with the corresponding positive definite dual slack matrix, which turns the SDP into a positive one even when $C$ was indefinite or negative definite (as for the two examples mentioned above).

\paragraph{Related Work} 
The study of Physarum dynamics were started in the mathematical biology community, in the context
of shortest path~\cite{nakagaki2007minimum}. Its theoretical foundations in computing has 
been studied in a line of works~\cite{Bonifaci2012shortest, Becchetti2013, miyaji2008physarum}.
Other works extended the Physarum dynamics to graph and network problem beyond shortest path,
including the design of transportation networks~\cite{tero2010rules, watanabe2011traffic, yang2019bio},
supply-chain networks~\cite{zhang2017physarum}, and flow problems~\cite{straszak2016natural, bonifaci2022physarum}.
Going beyond the graph setting, \cite{johannson2012slime, straszak2015natural} showed that the directed 
Physarum dynamics can solve positive linear programs in standard form. Physarum dynamics are also 
studied in other more abstract optimization problems, including the basis pursuit problem~\cite{facca2018physarum,StraszakV22}.

Many works benefit from modifying the Physarum dynamics or studying them in an abstract setting.
In~\cite{bonifaci2017revised} the author revised the model of~\cite{tero2007mathematical} by considering a different controlling variable for the adaptation mechanism, and gave a generalized Physarum dynamic which subsumes several Physarum dynamics as special cases, including~\cite{KARRENBAUER2020260} and the two-norm dynamics~\cite{bonifaci2022physarum}. In~\cite{KARRENBAUER2020260} they analyzed the convergence of ``non-uniform'' variant of the Physarum dynamics. In~\cite{bonifaci2022physarum} two different Physarum dynamics were proposed and analyzed for the multi-commodity flow problem.
By studying a Meta-Algorithm, \cite{StraszakV22} established a connection between the Iteratively Reweighted Least
Squares (IRLS) algorithm and the Physarum dynamics.

\paragraph{Our contribution}
In this paper, we generalize the \dynamicLP{} to SDP, by presenting 
a general recipe for the \dynamicSDP. We show that while being general enough,
our framework already guarantees some very desirable properties.
We propose a natural generalization of the LP Physarum conductance matrix and hence obtain 
our \dynamicOne. Key technical challenges are identified in the analysis
of \dynamicOne's soundness and convergence to equilibrium, under mild 
conditions. We present strong empirical evidence for the dynamic's unconditional convergence 
to optimality. With our understanding of the technical challenges posed by 
the \dynamicOne, and following our general recipe, we propose a
second conductance matrix and hence obtain our \dynamicTwo. 
We prove unconditional guarantee on the \dynamicTwo's soundness,
and more importantly, convergence to optimality. We believe that this showcases 
the potential and power of our general framework for \dynamicSDP. Although our implementation was geared towards the investigation of soundness and convergence instead of optimizing for performance, we are convinced that numerical algorithms obtained from carefully discretizing the Physarum dynamics motivate further research into practical SDP solvers based on nature-inspired algorithms.

\paragraph{Organization of the paper} 
In section~\ref{sec:overview} we give an overview of our contributions and explain 
the general framework and some of the key results. In section~\ref{sec:general-recipe}
we investigate our general recipe of \dynamicSDP{} and prove its desirable 
properties. Then in section~\ref{sec:1-dynamic} we propose a natural generalization of 
the \dynamicLP's conductance matrix and hence give our \dynamicOne.
In section~\ref{sec:1-soundness} and~\ref{sec:1-convergence} we analyse the soundness and 
covnergence properties of the \dynamicOne, under mild conditions.
Section~\ref{sec:2-dynamic} proposes the \dynamicTwo{} and proves 
unconditional soundness and convergence-to-optimality results.
We also highlight key differences between the two Physarum dynamics, and the technical
advantages of the \dynamicTwo{} which enables us to prove the strong 
convergence results. Section~\ref{sec:algo-experiment} presents two numerical algorithms 
obtained from discretizing the \dynamicSDP. In section~\ref{sec:experiment}
we discuss the experimental evaluations of the two numerical algorithms, providing 
empirical evidence for the unconditional soundness and convergence of 
the \dynamicOne. We provide practical implementations 
of the solvers and demonstrate their applications in various discrete optimization problems.

\section{Overview}\label{sec:overview}

In this section we give an overview of the \dynamicSDPs{} that are 
studied in this paper. The \dynamicSDPs{} are naturally inspired by 
previous works in \dynamicLP. These dynamics are charactarized
by the choices of conductance matrices in the system of ordinary differential equations describing the evolution, and we 
give a concrete choice of the conductance, inspired by and closely related to the \dynamicLP. However, the previous tehniques
do not work out of the box, due to the noncommutative nature of matrix multiplication. We show 
that the \dynamicSDP{} is sound and convergent, under mild conditions
on the objective. We then move on to show that a slight variation on the choice of the 
conductance matrix enables us to prove the soundness and convergence-to-optimality of the \dynamicSDP{} without these mild assumptions. 

For an easier comparison with the LP case note that one can also write an SDP in a vectorized notation where 
$\vectorize{M}$ for an $n \times n$ matrix $M$ is an $n^2$ vector whose $n(i-1) + j$th element is $M_{i,j}$, 
i.e., it is obtained by stacking up the columns of $M$ on top of each other.
Using this notation, we may write $tr(A^T B) = \sum_{i,j} a_{ij} b_{ij} = \vectorize{A}^T \vectorize{B}$.
Consider $\capitalize{A}$ to be an $m \times n^2$ matrix defined as

\begin{equation} \label{eq:omega-definition}
\capitalize{A} = \begin{bmatrix}
\vectorize{A_1}^T \\ \vectorize{A_2}^T \\ ... \\ \vectorize{A_m}^T
\end{bmatrix}.
\end{equation}

Now we can also write SDP in the vectorized format, which resembles that of the LP case:

\begin{equation*}
\min \{\vectorize{C}^T \vectorize{X} ~ : ~ \capitalize{A} \cdot \vectorize{X} = b,~ X \succeq 0\} =  \max \{ b^T y : \capitalize{A}^T y + \vectorize{S} = \vectorize{C}, \, S \succeq 0 \} \> \text{for }C \succ 0 
\end{equation*}
where we use the convention $\min \emptyset = \infty$ in the case of primal infeasibility, which is equivalent to dual unboundedness because $C \succ 0$ implies dual feasibility.

\paragraph{A general recipe for \dynamicSDPs} 
A natural way to extend the \dynamicLP~\eqref{eq:physarum-dynamics-lp-proj} to 
the SDP case is to lift the differential equations to matrices. For simplicity, here we 
assume that the \dynamicSDP{} start with a feasible point $X(0)$. 
In particular, we say that a point $X$ is \textbf{linearly feasible} if it satisfies
all linear constraints.
See section~\ref{sec:1-infeasible} for a detailed explanation on 
how to accommodate infeasible starting points. Our \dynamicSDP{} is defined 
by the following differential equation:
\begin{equation}
    \label{eq:physarum-dynamics-sdp-general}
    \vectorize{\dot{X}} = - (I - G\capitalize{A}^T(\capitalize{A}G\capitalize{A}^T)^\dagger \capitalize{A})
        G\vectorize{C},
\end{equation}
where, again, $G$ is some \textbf{conductance} matrix that characterizes the \dynamicSDPs.
While the dynamic itself is straightforward to generalize, the choice of conductance is not as obvious. 
A trivial generalization of the conductance of \dynamicLP{} would be the $n^2\times n^2$
matrix with the diagonal stacking all entries of $C^{-1}X$ while the other entries zero. This, 
however, will not work because this choice of conductance directly yields the \dynamicLP{} for the problem
\[
	\min \left \{ \vectorize{C}^T\vectorize{X}:\capitalize{A}\vectorize{X}=b, 
		\vectorize{X}\geq 0 \right \},
\]
which does not agree with the SDP problem because a positive semi-definite matrix may have negative entries. Moreover, for $\vectorize{C}$, $\vectorize{X}(0)$ with non-negative entries, the dynamics with that conductance matrix will keep all entries of $\vectorize{X}$ non-negative, and for negative entries, convergence would not be guaranteed at all.

On the other hand, using the sole requirement $G \succeq 0$ without any further specification of its structure, 
our general recipe can already be shown to reduce the objective values monotonically over time.
Writing the objective value along the trajectory of the dynamic as a function over time as
\[
    \mathcal{L}(t) \coloneqq tr(CX(t)) = \vectorize{C}^T \vectorize{X}(t),
\]
we obtain the following Theorem.

\begin{theorem}
    \label{thm:dynamic-reduce-objective}
    Using any symmetric positive semi-definite conductance $G$ for the \dynamicSDP~\eqref{eq:physarum-dynamics-sdp-general}, 
	then the following is true:
    \begin{enumerate}
        \item $\frac{d}{dt}\mathcal{L}(t)\leq 0$
        \item $\frac{d}{dt}\mathcal{L}(t)$ becomes zero if and only if $\dot{X}(t)=0$
        \item If $\mathcal{L}(t)$ is bounded from below, the dynamic converges to equilibrium, i.e., $\lim_{t\to \infty}\norm{\dot{X}(t)}=0$
    \end{enumerate}
\end{theorem}
We restate and prove this theorem as Theorem~\ref{thm:dynamic-reduce-objective-restate}.
This theorem shows that while our general recipe is general enough to accomodate
essentially any positive semi-definite conductance matrix, it still offers 
guarantees on how the objective values behave along the dynamics' trajectories.

\paragraph{\dynamicOneCap} Following the general recipe, we give a natural generalization 
of the conductance matrix employed in the \dynamicLP~\eqref{eq:physarum-dynamics-lp-proj}, and 
hence get our first \dynamicSDP, which we call the \emph{\dynamicOne}. We also make a statement about having infeasible starting point, see section~\ref{sec:1-infeasible} for how we slightly modify the dynamic to cope with infeasible starting point.
Under mild assumption on the objective $C$ of the SDP, we obtain the following results on the 
\dynamicOne:
\begin{theorem}{(Informal)}
    \label{thm:1-sdp-dynamic-informal}
    For the class of positive SDPs where $C^{-1}$ is linearly feasible,
    our \dynamicOne{} $G\coloneqq \frac{1}{2}(C^{-1} \otimes X + X\otimes C^{-1} )$ for the conductance in~\eqref{eq:physarum-dynamics-sdp-general} satisfies the following:
    \begin{enumerate}
        \item the dynamic stays linearly feasible,
        \item the dynamic is sound, i.e., $X(T)\succ 0$ for any finite time $T\geq 0$.
        \item the dynamic reduces the objective and reaches equilibrium, in particular,
            $\frac{d}{dt}tr(CX(t))\leq 0$ where equality holds if only if $\dot{X}(t)=0$. Moreover, 
            $\lim_{t\to \infty}\norm{\dot{X}(t)}=0$.
    \end{enumerate}
\end{theorem}
If the equilibrium points are positive definite, then statement 3 implies that the  \dynamicOne{} converges to optimality. In section~\ref{sec:experiment} we also provide empirical 
evidences that the dynamic converges to optimality without such assumption.
In section~\ref{sec:augmentation}
we introduce the augmented SDP such that the linear feasibility of $C^{-1}$ can be satisfied.
The first statement of Theorem~\ref{thm:1-sdp-dynamic-informal} is formally stated and proved as 
Theorem~\ref{thm:x-stays-feasible}. The second statement is formally stated and proved as Theorem~\ref{thm:1-dynamic-pd-cone}. The third statement is Corollary~\ref{cor:1-dynamic-convergence}
implied by Theorem~\ref{thm:dynamic-reduce-objective}.

\paragraph{\dynamicTwoCap} The soundness proof of 
Theorem~\ref{thm:1-sdp-dynamic-informal} relies on a mild assumption on the objective 
matrix $C$. Even though we can augment the problem so that the assumption is satisfied and 
that the conductance in the \dynamicOne{} is a natural extension of 
the conductance in~\eqref{eq:physarum-dynamics-lp-proj}, we study a different choice of 
conductance, namely $G\coloneqq X\otimes X$ and obtain our \emph{\dynamicTwo}.
With the \dynamicTwo, we are able to remove the mild assumption 
in~\ref{thm:1-sdp-dynamic-informal}. We summarize it as the following:

\begin{theorem}{(Informal)}
    \label{thm:2-sdp-dynamic-informal}
    For the class of positive SDPs, our \dynamicTwo{} $G \coloneqq X\otimes X$ for the conductance in~\eqref{eq:physarum-dynamics-sdp-general} satisfies the following:
    \begin{enumerate}
        \item the dynamic stays linearly feasible,
        \item the dynamic is sound, i.e., $X(T)\succ 0$ and for any finite time $T\geq 0$,
        \item the dynamic reduces the objective and reaches equilibrium, in particular,
        $\frac{d}{dt}tr(CX(t))\leq 0$ where equality holds if only if $\dot{X}(t)=0$. Moreover, 
        $\lim_{t\to \infty}\norm{\dot{X}(t)}=0$. 
    \end{enumerate}
\end{theorem}
Similar to before, if the equilibrium points are positive definite, 
then statement 3 implies that the \dynamicTwo{} converges to optimality.
The first statement of Theorem~\ref{thm:2-sdp-dynamic-informal} is formally stated and proved 
as Theorem~\ref{thm:x-stays-feasible}. The second statement is formally stated and proved as Theorem~\ref{thm:2-dynamic-pd-cone} and the third statement as Corollary~\ref{cor:2-dynamic-convergence}.
However, with the \dynamicTwo, we are able to avoid the technical obstacles emerging 
from analyzing the \dynamicOne, and prove results that are much stronger: we show that 
the \dynamicTwo{} coincide with a central path of SDP~\eqref{eq:SDP-definition}
when starting with a strictly feasible point, implying that the dynamic converges to optimality
unconditionally. We believe that this shows the great power and potential in our general recipe:
by cleverly designing the conductance, one can design Physarum dynamics that have strong 
theoretical guarantees. We state this result informally as the following:
\begin{theorem}{(Informal)}
	\label{thm:2-dynamic-central-path-informal}
	Fix any feasible point $F\succ 0$, there exists a central path of SDP~\eqref{eq:SDP-definition}
	that coincides with the \dynamicTwo{} starting from $F$.
	Consequently, the \dynamicTwo{} converges to the optimum of 
	SDP~\eqref{eq:SDP-definition} when starting from any feasible 
	point $X(0)\succ 0$.
\end{theorem}
We discuss this result thoroughly in section~\ref{sec:2-convergence} and formally state and prove 
Theorem~\ref{thm:2-dynamic-central-path-informal} as Theorem~\ref{thm:2-dynamic-central-path}.

\section{A General Recipe for \dynamicSDPCaps}
\label{sec:general-recipe}
In section~\ref{sec:overview} we briefly explained that our general recipe for \dynamicSDP~\eqref{eq:physarum-dynamics-sdp-general} is a natural extension of the \dynamicLP{} to SDP. One might ask, why do we need a general framework of \dynamicSDPs{} that admit any choices of positive semidefinite conductance $G$? As it turns out, 
given any concrete choice of conductance matrix $G$ that is positive semi-definite,
the \dynamicSDP{} already satisfies two important properties:
starting with a linearly feasible $X(0)$,
\begin{enumerate}
	\item the dynamic stays linearly feasible.
	\item the objective value is reduced along the trajectory of the dynamic.
\end{enumerate}
This means that once a positive semidefinite conductance matrix $G$ is chosen such 
that the dynamic stays within the positive definite cone, an \dynamicSDP{} that 
stays feasible and reduces objective is found.

In this section, we prove these two important properties of the general recipe of 
SDP Phyarum dynamic~\eqref{eq:physarum-dynamics-sdp-general}.

\paragraph{Linear feasibility} Geometrically speaking, given any positive definite
conductance $G$, the matrix 
$-(I - G\capitalize{A}^T(\capitalize{A}G\capitalize{A}^T)^{-1} \capitalize{A})$ is the 
projection into the kernel of $\capitalize{A}$, under the inner product:
\[
    \langle a, b\rangle_{G^{-1} } \coloneqq a^T G^{-1}  b
\]
We will expand on this in section~\ref{sec:1-infeasible} and obtain an equivalent formulation 
of the \dynamicSDPs. Now, what this geometric perspective entails for us is that 
in our general recipe of the SDP Phsarum dynamics, $\vectorize{\dot{X}}(t)$ is always in the 
kernel of $\capitalize{A}$.
\begin{lemma}
    \label{lem:xdot-in-kernel}
    Given any positive semi-definite conductance $G$,
    $-(I - G\capitalize{A}^T(\capitalize{A}G\capitalize{A}^T)^\dagger \capitalize{A})G\vectorize{C}$ is in the 
    kernel of $\capitalize{A}$. In particular, given the \dynamicSDP~\eqref{eq:physarum-dynamics-sdp-general}, for any time $t$,
    \[
        \capitalize{A}\vectorize{\dot{X}}(t) = 0
    \]
\end{lemma}
\begin{proof}
    Let $\Tilde{\capitalize{A}}\coloneqq \capitalize{A}G^\frac{1}{2}$, there exists $v$ and $d$ such that 
	\[
		G^{\frac{1}{2}}\vectorize{C} = \Tilde{\capitalize{A}}^Tv+d
	\]
	where $\Tilde{\capitalize{A}} d =0$. Now we can verify that:
    \begin{align*}
        \capitalize{A}\vectorize{\dot{X}}(t)&= -\capitalize{A}(I - G\capitalize{A}^T
			(\capitalize{A}G\capitalize{A}^T)^\dagger \capitalize{A})G\vectorize{C}\\
        &= -\Tilde{\capitalize{A}}G^\frac{1}{2}\vectorize{C} + \Tilde{\capitalize{A}}\Tilde{\capitalize{A}}^T
			(\Tilde{\capitalize{A}}\Tilde{\capitalize{A}}^T)^\dagger \Tilde{\capitalize{A}} 
			G^\frac{1}{2}\vectorize{C}\\
		&= -\Tilde{\capitalize{A}}\Tilde{\capitalize{A}}^Tv -\Tilde{\capitalize{A}}d
			+ \Tilde{\capitalize{A}}\Tilde{\capitalize{A}}^T
			(\Tilde{\capitalize{A}}\Tilde{\capitalize{A}}^T)^\dagger \Tilde{\capitalize{A}} 
			\Tilde{\capitalize{A}}^Tv + \Tilde{\capitalize{A}}\Tilde{\capitalize{A}}^T
			(\Tilde{\capitalize{A}}\Tilde{\capitalize{A}}^T)^\dagger \Tilde{\capitalize{A}} d\\
        &= -\Tilde{\capitalize{A}}\Tilde{\capitalize{A}}^Tv
			+\Tilde{\capitalize{A}}\Tilde{\capitalize{A}}^T
			(\Tilde{\capitalize{A}}\Tilde{\capitalize{A}}^T)^\dagger \Tilde{\capitalize{A}} 
			\Tilde{\capitalize{A}}^Tv\\
		&= 0
    \end{align*}
    So indeed $\vectorize{\dot{X}}(t)$ is in the kernel of $\capitalize{A}$. 
\end{proof}
This leads to our next theorem that shows $X(t)$ stays linearly feasible, if the \dynamicSDP{} 
starts with a linearly feasible $X(0)$.
\begin{theorem}
    \label{thm:x-stays-feasible}
    Given a linearly feasible starting point for the \dynamicSDPs~\eqref{eq:physarum-dynamics-sdp-general},
	and a positive semi-definite conductance $G$,
    then at any time $t$, 
    \[
        \capitalize{A}\vectorize{X}(t) = b
    \]
    In other words, $X(t)$ stays feasible.
\end{theorem}
\begin{proof}
    First note that:
    \[
        \frac{d}{dt} \capitalize{A}\vectorize{X}(t) - b = \capitalize{A}\vectorize{\dot{X}}(t)
    \]
    Since $X(0)$ is linearly feasible, $\capitalize{A}\vectorize{X}(0) - b=0$. By fundamental theorem of calculus, we have:
    \begin{align*}
        \capitalize{A}\vectorize{X}(t) - b & = \capitalize{A}\vectorize{X}(0) - b + \int_0^t\capitalize{A}\vectorize{\dot{X}}(s)ds\\
            &= 0
    \end{align*}
    where in the last equality we applied Lemma~\ref{lem:xdot-in-kernel}. Therefore,
    $X(t)$ stays feasible.
\end{proof}

\paragraph{Objective value along the dynamic's trajectory} Now we move on to show 
that with any positive semi-definite conductance matrix $G$, our general recipe 
of \dynamicSDP{} reduces the objective value. We restate
Theorem~\ref{thm:dynamic-reduce-objective} and give a formal proof:
\begin{theorem}
    \label{thm:dynamic-reduce-objective-restate}
    Using any symmetric positive semi-definite conductance $G$ for the \dynamicSDP~\eqref{eq:physarum-dynamics-sdp-general},  
	then the following is true:
    \begin{enumerate}
        \item $\frac{d}{dt}\mathcal{L}(t)\leq 0$
        \item $\frac{d}{dt}\mathcal{L}(t)$ becomes zero if and only if $\dot{X}(t)=0$
        \item If $\mathcal{L}(t)$ is bounded from below, the dynamic converges to equilibrium $\lim_{t\to \infty}\norm{\dot{X}(t)}=0$
    \end{enumerate}
\end{theorem}

\begin{proof}
    It's easy to see that $\frac{d}{dt}\mathcal{L}(t)=tr(C\dot{X}(t))$. By the definition 
    of the \dynamicSDP, we have that:
    \[
        \frac{d}{dt}\mathcal{L}(t) = -\vectorize{C}^TG\vectorize{C}+\vectorize{C}^TG
            \capitalize{A}^T(\capitalize{A}G\capitalize{A}^T)^\dagger \capitalize{A}G\vectorize{C}
    \]
    \begin{enumerate}
        \item Write $P_G\coloneqq G^{\frac{1}{2}}\capitalize{A}^T(\capitalize{A}G\capitalize{A}^T)^\dagger 
            \capitalize{A}G^{\frac{1}{2}}$. It's easy to see that $P_G$ is symmetric. Moreover,
            $P_G^2=G^{\frac{1}{2}}\capitalize{A}^T(\capitalize{A}G\capitalize{A}^T)^\dagger 
            \capitalize{A}G^{\frac{1}{2}}G^{\frac{1}{2}}\capitalize{A}^T(\capitalize{A}G\capitalize{A}^T)^\dagger 
            \capitalize{A}G^{\frac{1}{2}}=G^{\frac{1}{2}}\capitalize{A}^T(\capitalize{A}G\capitalize{A}^T)^\dagger 
            \capitalize{A}G^{\frac{1}{2}}$, i.e. $P_G^2=P_G$. Therefore, $P_G$ is an orthogonal projection.
            
            In particular, for any $v, \norm{v}_2\geq \norm{P_G v}_2$, where equality holds 
			if and only if $v$ is an eigenvector of $P_G$.
            Hence,
            \begin{align*}
                \frac{d}{dt}\mathcal{L}(t) &= -\vectorize{C}^TG\vectorize{C}+ \vectorize{C}^TG^{\frac{1}{2}}
                    P_GG^{\frac{1}{2}}\vectorize{C}\\
					&= -\vectorize{C}^TG\vectorize{C}+ \vectorize{C}^TG^{\frac{1}{2}}
                    P_GP_GG^{\frac{1}{2}}\vectorize{C}\\
					&= -\vectorize{C}^TG\vectorize{C}+\norm{P_GG^{\frac{1}{2}}\vectorize{C}}_2^2\\
                    & \le 0,
            \end{align*}
            where in the inequality we used the fact that $P_G$ is an orthogonal projection.
        \item Due to the proof above, $\frac{d}{dt}\mathcal{L}(t)$ becomes zero if only if the inequality
            is an equality in the previous proof. In particular, this implies that $G^{\frac{1}{2}}\vectorize{C}$ 
			is an eigenvector of the projection
            $P_G$. Note that $P_G$ only has two eigenvalues, 0 and 1. If $P_GG^{\frac{1}{2}}\vectorize{C}=0$,
            then by definition $\vectorize{\dot{X}}(t)=-G\vectorize{C}$. However, in this case we also have 
            that 
            \[
                \frac{d}{dt}\mathcal{L}(t) = -\vectorize{C}^TG\vectorize{C}= -\norm{G^{\frac{1}{2}}\vectorize{C}}_2^2
            \]
            Therefore $G^{\frac{1}{2}}\vectorize{C}$ must be zero as well. This implies that 
            $\vectorize{\dot{X}}(t)=0$. On the other hand, if 
            $P_GG^{\frac{1}{2}}\vectorize{C}=G^{\frac{1}{2}}\vectorize{C}$, then
            \begin{align*}
                \vectorize{\dot{X}}(t) &= -G\vectorize{C}+G^{\frac{1}{2}}P_GG^{\frac{1}{2}}\vectorize{C}\\
                    & = -G\vectorize{C} + G^{\frac{1}{2}}G^{\frac{1}{2}}\vectorize{C}\\
                    & = 0
            \end{align*}
            \item We know that $\mathcal{L}$ at time $0$ is finite and its derivative is negative. The derivative
                remains negative as long as $\dot{X}\neq 0$. In addition, $\mathcal{L}(t)=tr(CX(t))$ is lower 
                bounded from below. Therefore, $\dot{X}$ must eventually converge to $0$.
    \end{enumerate}
\end{proof}

\section{\dynamicOneCap}\label{sec:1-dynamic}

In section~\ref{sec:overview} we introduced a general recipe for \dynamicSDP~\eqref{eq:physarum-dynamics-sdp-general}. 
In this section we study the soundness and convergence of the \dynamicOne.
For most part of this section,
we focus on the case where the dynamic has feasible starting point. 
We address the extension of the \dynamicOne{}
to infeasible starting points in section~\ref{sec:1-infeasible}.
It remains unclear what is 
a natural generalization of the conductance matrix $C^{-1} X$ in the \dynamicLP~\eqref{eq:physarum-dynamics-lp-proj}. We argue that by setting 
\[
	G\coloneqq \frac{1}{2}(C^{-1}\otimes X + X\otimes C^{-1}),
\]
the \dynamicSDP{} reduces to the \dynamicLP{} in the diagonal 
setting. Without going over all the details, we note that the key point here is 
to observe that, when $C,X$ and $A_\ell$'s are diagonal, then the $(k,\ell)$-th 
entry of $\capitalize{A}G\capitalize{A}^T$ is $tr(A_kXA_\ell C^{-1})$, which equals 
to $a_k^TC^{-1}Xa_\ell$ where $a_k$ and $a_\ell$ are the LP linear constraints.
In other words, $\capitalize{A}G\capitalize{A}^T=ACX^{-1}A^T$ where $A$ is the 
LP constraint matrix. Also note that $\capitalize{A}G\vectorize{C}=b$.
Observe that $\frac{1}{2}(C^{-1}\otimes X + X\otimes C^{-1})$ is 
positive (semi-)definite when $X$ is positive (semi-)definite.

\subsection{Soundness of the \dynamicOneCap}
\label{sec:1-soundness}
In this section we show that, under mild conditions, if $X(t)$ is a matrix 
funciton defined by~\eqref{eq:physarum-dynamics-sdp-general} with conductance matrix 
$G=\frac{1}{2}(C^{-1} \otimes X + X\otimes C^{-1} )$, then in any finite time,
$X(t)$ stays feasible. Two types of conditions comprise the feasibility of $X(t)$ in the SDP:
\begin{itemize}
    \item Linear constraints of the form $tr(A_\ell X(t)) = b_\ell$ which can be summarized 
        in $\capitalize{A} \cdot \vectorize{X} = b$.
    \item Positive semi-definiteness of $X(t)$, i.e., $X(t) \succeq 0$.
\end{itemize}

In section~\ref{sec:general-recipe} we have shown that, as long as $G$ stays positive 
semi-definite, our dynamic stays linearly feasible. Now it remains to show 
that $X(t)$ stays within the positive definite cone in any finte time.

\paragraph{Positive definiteness} Proving that the Physarum dynamics stays within the 
positive definite cone is harder than proving linear feasibility. Typically one employs 
a \textbf{work function} type argument and shows that the 
work function does not explode to infinity. In~\cite{straszak2015natural} the work function 
for \dynamicLP{} is defined to be 
$\mathcal{B}(t) \coloneqq \sum_{i\in [n]}y_ic_i\ln x_i(t)$
where $y$ is some feasible solution to the LP and $c$ is the objective coefficient and $x(t)$
is the \dynamicLP~\eqref{eq:physarum-dynamics-lp-proj}. This is an entropy type 
work function, and the $c$ and $y$ terms are introduced to counter the $C^{-1} $ component 
in the conductance $C^{-1} X$. A crucial distinction between the SDP problem and the LP problem 
is that matrix multiplications are in general non-commutative. Consequently, differentiating 
a entropy type work function on matrices leads to very complicated expressions in terms 
of series. Handling these series would eventually require some commutativity assumption 
along the trajectory of the dynamic, which is unrealistic. 
Given the general recipe of \dynamicSDPs~\eqref{eq:physarum-dynamics-sdp-general}
and the choice of conductance $G=\frac{1}{2}(C^{-1} \otimes X + X\otimes C^{-1} )$, we define 
the following work function:
\begin{equation}
    \label{eq:work-function-definition}
    W(t) \coloneqq \ln \det (X(t))
\end{equation}
This is also an adaptation of a well-known barrier function for the positive definite cone. 
We establish the following lemma connecting the positive definiteness of $X(t)$ with the 
derivative of the work function $W(t)$.
\begin{lemma}
    \label{lem:pd-work-function-diff}
    Consider any dynamics $X(t)$ with a starting point $X(0)\succ 0$.
    If there exists some constant $\mu$ (independent of the time $t$) such that 
    \[
        \frac{d}{dt} W(t) \geq \mu,
    \]
    then for any finite time $T\geq 0$, we have $X(T)\succ 0$
\end{lemma} 
\begin{proof}
    If at any time $X(t)\not\succ 0$, then by the assumption that $X(0)\succ 0$ and the 
    continuity of the eigenvalues of $X(t)$, there must exists a finite time 
    when $X(t)$ obtains a zero eigenvalue. Consider the first time $T> 0$ where 
    $X(T)$ obtains a zero eigenvalue. Now we have
    \begin{align*}
        \lim_{t\to T^-}W(t) &= W(0)+\lim_{t\to T^-}\int_{s=0}^t \frac{d}{ds} W(s) ds \\
            & \geq W(0) +\mu T
    \end{align*}
    However, by our assumption that $X(T)$ obtains a zero eigenvalue, $\lim_{t\to T^-}W(t)$
    must be divergent and not lower bounded by any finite value. This is a contradoction
    and we get that for any finite time $T\geq 0, X(T)\succ 0$.
\end{proof}
With Lemma~\ref{lem:pd-work-function-diff}, showing that our \dynamicOne,
or indeed, any concrete instantiation of our general recipe of \dynamicSDPs,
stays within the positive definite cone reduces to showing that the work function's derivative
is lower bounded. This is where we have to introduce an additional condition on the objective 
of the SDP. We require that $C^{-1} $ is linearly feasible, in other words, 
$\capitalize{A}\vectorize{C^{-1} }=b$. The reason is that in the \dynamicLP's work 
function $\mathcal{B}(t)$, the entropy-like nature corrects and replaces the $C^{-1} $ terms
showing up in the derivative. Some may argue that requiring $C^{-1} $ to be linearly feasible
is too restrictive. However, in section~\ref{sec:augmentation}, we present how to 
augment the problem so that $C^{-1} $ is always linearly feasible. Moreover,
in section~\ref{sec:experiment}, we show strong empirical evidences that 
the dynamic stays within the positive definite cone regardless of this mild assumption.
In section~\ref{sec:2-dynamic} we show that a different choice of the conductance matrix 
will remove this condition on $C^{-1} $. Now, we lower bound the derivative
of the work function for our 
\dynamicOne, and hence prove that the dynamic stays within the 
positive definite cone.
\begin{theorem}
    \label{thm:1-dynamic-pd-cone}
    If $C^{-1} $ and $X(0)$ are linearly feasible, then for the \dynamicOne{} where 
    the conductance is $G=\frac{1}{2}(C^{-1} \otimes X + X\otimes C^{-1} )$, when $X\succ 0$, the derivative 
    of the work function $W(t)=\ln \det (X(t))$ is lower bounded by a time-independent
    constant. Furthermore, by Lemma~\ref{lem:pd-work-function-diff}, the \dynamicOne{} stays in the positive definite cone at any finite time.
\end{theorem}
\begin{proof}
    First we compute the derivative of $W(t)$~\cite{petersen2008matrix}
    \[
        \frac{d}{dt}W(t) = tr(X^{-1} (t)\dot{X}(t))
    \]
    Note that $tr(X^{-1} \dot{X}(t))=\vectorize{X^{-1} }^T\vectorize{\dot{X}}(t)$,
    so by the definition of the \dynamicSDPs~\eqref{eq:physarum-dynamics-sdp-general},
    and that $G=\frac{1}{2}(C^{-1} \otimes X + X\otimes C^{-1} )$, we have
    \begin{align*}
        \frac{d}{dt}W(t) &= \vectorize{X^{-1} }^T\vectorize{\dot{X}}(t)\\
            &= - \vectorize{X^{-1} }^T(I - G\capitalize{A}^T(\capitalize{A}G\capitalize{A}^T)^{-1} \capitalize{A})
            G\vectorize{C}\\
            &= - \vectorize{X^{-1} }^TG\vectorize{C} + \vectorize{X^{-1} }^TG\capitalize{A}^T
                (\capitalize{A}G\capitalize{A}^T)^{-1} \capitalize{A}G\vectorize{C}
    \end{align*}
    Note that~\footnote{We frequently use the rule: 
        $\vectorize{ABC} = (C^T \otimes A) \vectorize{B}$.} $G\vectorize{X^{-1} }=\frac{1}{2}((C^{-1} \otimes X)\vectorize{X^{-1} } +
        (X\otimes C^{-1} )\vectorize{X^{-1} })=\vectorize{C^{-1} }$, and similarly
        $G\vectorize{C}=\vectorize{X}$. Hence we have
    \begin{align*}
        \frac{d}{dt}W(t) &= - \vectorize{C^{-1} }^T\vectorize{C} + \vectorize{C^{-1} }^T\capitalize{A}^T
            (\capitalize{A}G\capitalize{A}^T)^{-1} \capitalize{A}G\vectorize{C}\\
            &= -n + \vectorize{C^{-1} }\capitalize{A}^T(\capitalize{A}G\capitalize{A}^T)^{-1} \capitalize{A}
                \vectorize{X}\\
            &= -n +b^T(\capitalize{A}G\capitalize{A}^T)^{-1} b\\
            &\geq -n
    \end{align*}
    where the third equality uses the condition that $C^{-1} $ and $X(0)$ are linear feasible,
    and the inequality uses the fact that $\capitalize{A}G\capitalize{A}^T$, and therefore 
    $(\capitalize{A}G\capitalize{A}^T)^{-1} $, is positive definite.
\end{proof}
\begin{remark}
    In fact, based on the proof above, it suffices that any positive multiple of $C^{-1}$ is linearly feasible. Note that scaling the objective does not change the solution of the problem.
\end{remark}
This concludes that, under the assumption that $C^{-1} $ is linearly feasible, and with 
a feasible starting point $X(0)\succ 0$
the \dynamicOne{} is sound, in other word, the dynamic stays 
linearly feasible and within the positive definite cone.

\subsection{Convergence of the \dynamicOneCap}
\label{sec:1-convergence}
In this section we consider the convergence of the \dynamicOneCap.
In particular, we study how objective value behaves along the trajectory of the Physarum
dynamic. 
Since we have established in section~\ref{sec:1-soundness} that our \dynamicOne{} is sound, under the condition that $C^{-1} $ is linearly feasible, 
we have the immediate corollary of Theorem~\ref{thm:dynamic-reduce-objective}:

\begin{corollary}
    \label{cor:1-dynamic-convergence}
    Given a feasible starting point $X(0)\succ 0$, and that $C^{-1} $ is linearly feasible, 
    then \dynamicOne{} satisfies the following:
    \begin{enumerate}
        \item $\frac{d}{dt}\mathcal{L}(t)\leq 0$
        \item $\frac{d}{dt}\mathcal{L}(t)$ becomes zero if and only if $\dot{X}(t)=0$
        \item The dynamic converges to equilibrium $\lim_{t\to \infty}\norm{\dot{X}(t)}=0$
    \end{enumerate}
\end{corollary}

This corollary does not immediately imply the dynamic's convergence to optimality.
The dynamic might not converge to a single point, but multiple equilibrium points.
We define the following equilibrium set to charactarize them:
\begin{equation}
    \label{eq:1-eqset-definition}
    EQ_1 \coloneqq \{X: v_1(\vectorize{X}) = 0, \capitalize{A} \vectorize{X} = b, X \succeq 0 \}
\end{equation}

where we defined $v_1$ as the velocity corresponding to our autonomous \dynamicSDP. In other words, 
$$v_1(\vectorize{X}) = G \capitalize{A}^T (A G \capitalize{A}^T)^\dagger \capitalize{A} \vectorize{X} \text{ where } G = \frac{1}{2} \left(C^{-1} \otimes X + X \otimes C^{-1}\right).$$
Moreover, we used the Moore-Penrose pseudo inverse to accomodate the possibility that 
the dynamic might approach singular points at the limit. Now we further charactarize the positive definite
points in the equilibrium set $EQ_1$, and consequently deduce a conditional convergence-to-optimality
result on the \dynamicOne.

\begin{lemma}
    \label{lem:1-eq-pd-opt}
    If $X\in EQ_1$ and $X\succ 0$, then $X$ is an optimal solution to SDP~\eqref{eq:SDP-definition}.
\end{lemma}

\begin{proof}
    By definition~\eqref{eq:1-eqset-definition}, any point $X\in EQ_1$ is a feasible solution to 
    SDP~\eqref{eq:SDP-definition}. If $X\succ 0$, we prove its optimality using week duality.
    Note that $G$ is nonsingular if $X\succ 0$, therefore we have:
    \[
        G^{-1} \vectorize{X} = \capitalize{A}^T(\capitalize{A}G
        \capitalize{A}^T)^{-1} \capitalize{A}\vectorize{X}
    \]
    Write $y\coloneqq (\capitalize{A}G\capitalize{A}^T)^{-1} 
        \capitalize{A}\vectorize{X}$ and note that $G^{-1} \vectorize{X}
        = \vectorize{C}$, then we have that 
    \[
        \vectorize{C} = \capitalize{A}^Ty
    \]
    In particular, by the definition~\eqref{eq:SDP-definition}, $(y,0)$ is a dual feasible solution.
    Moreover,
    \begin{align*}
        \vectorize{C}^T\vectorize{X} &= \vectorize{C}^TG \capitalize{A}^Tp\\
            &= \vectorize{X}^T\capitalize{A}^Tp\\
            & = b^Tp
    \end{align*}
    where in the third equality we used the feasibily of $X$. This means that for the primal-dual
    pair $(X,y,0)$, the primal and dual objective conincides. Therefore by weak duality we have 
    that $X$ is an optimal solution to the SDP.
\end{proof}

With Corollary~\ref{cor:1-dynamic-convergence} and Lemma~\ref{lem:1-eq-pd-opt}, we can obtain the 
following conditional convergence-to-optimality result for our \dynamicOne.

\begin{theorem}
    \label{thm:1-dynamic-cond-opt}
    Starting with a feasible $X(0)\succ 0$, assuming that $C^{-1} $ is linearly feasible,
    if $\forall X\in EQ_1, X\succ 0$, then the \dynamicOne{} is sound 
    and converges to the optimum of SDP~\eqref{eq:SDP-definition}.
\end{theorem}

We conjecture that the \dynamicOne{} converges to optimality regardless
of the condition on $C$ and $EQ_1$, and we summarize the conjecture as Conjecture~\ref{conj:general-dynamic-convergence} in section~\ref{sec:1-infeasible}, which also captures the behavior of the dynamic with linearly infeasible starting point. 
We remark that in section~\ref{sec:experiment} we present strong empirical 
evidence supporting our conjecture, motivating furthur research into this Physarum dynamic
to confirm its soundness and convergence to optimality. In section~\ref{sec:2-dynamic}
we show that choosing a different conductance for our general recipe of \dynamicSDP{} suffices to guarantee both soundness and convergence to optimality, showcasing 
the power of our general framework.

\subsection{Extension to Infeasible Starting Point}
\label{sec:1-infeasible}
In the previous sections we discussed the soundness and convergence of our 
\dynamicOne{} with a feasible starting point. 
In this section, we show that it is possible to extend the dynamic to cope with
infeasible starting points.
We also introduce a different perspective on the definition of $\dot{X}(t)$ which 
will be the starting point of the extension of the dynamic.

\paragraph{Update problem} Geometrically speaking, with positive definite $G$,
$-(I - G\capitalize{A}^T(\capitalize{A}G\capitalize{A}^T)^{-1} \capitalize{A})$ 
is a projection into the kernel of $\capitalize{A}$,
under the inner product induced by $G^{-1} $. In particular, given any $w$,
we have the following:
\[
    -(I - G\capitalize{A}^T(\capitalize{A}G\capitalize{A}^T)^{-1} \capitalize{A})w
    =\text{argmin}_{y} \left \{ \norm{w-y}^2_{G^{-1} }: \capitalize{A}y=0 \right \}
\]
We will not prove this claim, but instead, we 
prove the following charactarization of $\dot{X}(t)$ as an update problem:

\begin{theorem}
    \label{thm:xdot-update-feasible}
    Consider the \dynamicSDP~\eqref{eq:physarum-dynamics-sdp-general}, 
    and positive semi-definite conductance $G$, then
    \begin{equation}
        \label{eq:xdot-opt-feasible}
        \dot{X} = \text{argmin}_{F}\left \{ \vectorize{C}^T\vectorize{F} +
         \frac{1}{2}\vectorize{F}^T G^\dagger  \vectorize{F}:
         \capitalize{A}\vectorize{F}=0\right \},
    \end{equation}
\end{theorem}
\begin{proof}
    We directly solve the minimization problem using the method of Lagrange multipliers.
    Using the Lagrange multiplier $p$, $\dot{X}$ is the solution to the following 
    system:
    \begin{equation}
        \label{eq:xdot-lagranage-feasible}
        \begin{array}{rcl}
            \vectorize{C}+G^\dagger \vectorize{\dot{X}} & = & \capitalize{A}^Tp \\
            \capitalize{A}\vectorize{\dot{X}} & = & 0
        \end{array}
    \end{equation}
    Multiplying both sides of the first equation in~\eqref{eq:xdot-lagranage-feasible} with 
    $\capitalize{A}G$, we have that:
    \[
        \capitalize{A}G\vectorize{C} = \capitalize{A}G\capitalize{A}^Tp
    \]
    Solving for $p$, we get 
    \[
        p = (\capitalize{A}G\capitalize{A}^T)^\dagger \capitalize{A}G\vectorize{C}
    \]
    Plug it back into the first equation, we see that 
    \[
        \vectorize{\dot{X}} = -(I - G\capitalize{A}^T(\capitalize{A}G\capitalize{A}^T)^\dagger \capitalize{A})
            G\vectorize{C},
    \]
    which coincide with our original definition of the \dynamicSDP.
\end{proof}

Introducing the update problem allows us to extend the \dynamicSDP{} naturally 
to accommodate the infeasible starting point. In general, our \dynamicOne{} 
dynamic can be defined as the following:

\begin{equation}
    \label{eq:xdot-opt-general}
        \dot{X} = \text{argmin}_{F}\left \{ \vectorize{C}^T\vectorize{F} +
         \frac{1}{2}\vectorize{F}^T G^\dagger  \vectorize{F}:
         \capitalize{A}\vectorize{F}=b-\capitalize{A}\vectorize{X}(t)\right \},
\end{equation}
where $G=\frac{1}{2}(C^{-1} \otimes X + X \otimes C^{-1} )$ is still the conductance matrix.
Similar to the proof of Theorem~\ref{thm:xdot-update-feasible}, we give without proof 
the solution to the general update problem, and write out the \dynamicSDP{} 
explicitly:

\begin{equation}
    \label{eq:xdot-general-definition}
    \dot{X}(t) = G\capitalize{A}^T(\capitalize{A}G\capitalize{A}^T)^\dagger b-\vectorize{X}(t)
\end{equation}

It should be easy to see that we can equivalently define $\dot{X}$ with an auxilary variable
$Q$ such that $\dot{X}(t)\coloneqq Q(t)-X(t)$ where 
\begin{equation}
    \label{eq:qupdate-opt}
    Q = \text{argmin}_F\left \{ \vectorize{F}^TG^\dagger \vectorize{F}:
        \capitalize{A}\vectorize{F}=b \right \}
\end{equation}

\paragraph{Approaching linear feasibility} When the Physarum starts with a feasible point,
Theorem~\ref{thm:x-stays-feasible} shows that the dynamic stays linearly feasible. We now 
show that with a linearly infeasible starting point, the dynamic~\eqref{eq:xdot-general-definition}
approaches linear feasibility exponentially. 

\begin{theorem}
    \label{thm:exponential-linear-feasibility}
    Given any $X(0)\succ 0$, for any linear 
    constraints $\ell$, we have:
    \[
        tr(A_\ell X(t))-b_\ell = e^{-t}(tr(A_\ell X(0))-b_\ell)
    \]
    In case $X(0)$ is linearly feasible, $X(t)$ remains linearly feasible and 
    the dynamic~\eqref{eq:xdot-opt-general} reduces to the dynamic~\eqref{eq:xdot-opt-feasible}.
\end{theorem}

\begin{proof}
    By definiton of the update problem~\eqref{eq:xdot-opt-general}, 
    $\capitalize{A}\vectorize{\dot{X}}(t) = b - \capitalize{A}\vectorize{X}(t)$.
    We can write the following differential equation:
    \begin{align*}
        \frac{d}{dt} \left ( tr(A_\ell X(t))-b_\ell \right ) &= tr(A_\ell \dot{X}(t))\\
         &= b_\ell - tr(A_\ell X(t))
    \end{align*}
    Therefore, the function $\Delta_\ell(t)\coloneqq tr(A_\ell X(t)) -b_\ell$ can be characterized
    by 
    \[
        \dot{\Delta_\ell(t)} = - \Delta_\ell(t)
    \]
    The solution to this differential equation is $\Delta_\ell(t)=e^{-t}\Delta_\ell(0)$, that 
    is, $X(t)$ converges to linear feasibility exponentially fast.

    In the case $X(0)$ is linearly feasible, $\Delta_\ell(0)=0$ and $\Delta_\ell(t)=0$, in other 
    words, $X(t)$ stays linearly feasible.
\end{proof}

\paragraph{Positive definiteness} Given Lemma~\ref{lem:pd-work-function-diff}, to show that 
the general \dynamicOne{} stays in the positive definite cone, we only need 
to show that the derivative of the work function $W(t)=\ln\det X(t)$ is lower bounded 
by some time-independent constant.

\begin{theorem}
    \label{thm:general-dynamic-pd-cone}
    If $C^{-1} $ is linearly feasible and $X(0)\succ 0$ (not necessarily linearly feasible), 
    then for the general \dynamicOne, when $X(t)\succ 0$,
    the derivative of the work function $W(t)=\ln\det X(t)$ is lower bounded by a time-independent
    constant. Furthermore, by Lemma~\ref{lem:pd-work-function-diff}, the dynamic stays in the positive 
    definite cone at any finite time.
\end{theorem}

\begin{proof}
    Following the proof of Theorem~\ref{thm:1-dynamic-pd-cone}, we will lower bound $tr(X^{-1} \dot{X})$.
    By definition~\eqref{eq:xdot-general-definition}, 
    \begin{align*}
        \vectorize{X^{-1} }^T\vectorize{\dot{X}} &= \vectorize{X^{-1} }^T
            G\capitalize{A}^T(\capitalize{A}G\capitalize{A}^T)^{-1} b -n \\
            &= \vectorize{C^{-1} }^T\capitalize{A}^T(\capitalize{A}G\capitalize{A}^T)^{-1} b -n
    \end{align*}
    By our assumption that $C^{-1} $ is linearly feasible, we have:
    \begin{align*}
        \vectorize{X^{-1} }^T\vectorize{\dot{X}} &= b^T(\capitalize{A}G\capitalize{A}^T)^{-1} b -n\\
            &\geq -n,
    \end{align*}
    where the inequality follows from the positive definiteness of $(\capitalize{A}G\capitalize{A}^T)^{-1} $.
    The rest of the proof is identical to that of Theorem~\ref{thm:1-dynamic-pd-cone}.
\end{proof}

\paragraph{Convergence} The convergence results in section~\ref{sec:1-convergence} is not so easily
generalized. The main obstacle is that even though the errors $\Delta_\ell$ approach zero 
exponentially fast, they still break the proofs where we use inequalities that are tight.
Nevertheless, we present the following conjecture.

\begin{conjecture}
    \label{conj:general-dynamic-convergence}
    If $X(0)\succ F$, for some feasible solution $F$, then the \dynamicOne{} stays in the positive definite cone and converges to the optimum of the SDP.
\end{conjecture}
In section~\ref{sec:experiment} we present empirical evidence that this 
conjecture is indeed true.

\section{\dynamicTwoCap} \label{sec:2-dynamic}
Section~\ref{sec:1-dynamic} gives a natural generalization of the conductance 
matrix in the \dynamicLP~\eqref{eq:physarum-dynamics-lp-proj}, 
but due to the intrinsic difficulties in matrix computations, we are only able to prove soundness for the subclass of positive SDPs where $C^{-1}$ or any positive multiple thereof satisfies the linear constraints. In addition, the current techniques only 
give conditional convergence-to-optimality for \dynamicOne.
 In this section, we show that by choosing the conductance
\[
    G \coloneqq X\otimes X,
\]
soundness is guaranteed for our general recipe of SDP 
Physarum dynamic~\eqref{eq:physarum-dynamics-sdp-general}.
Moreover, we prove that, unconditionally, this dynamic converges to \textbf{optimality}.
The resulting \dynamicSDP{} is referred to as \emph{\dynamicTwo}.
In this section we assume that the Physarum dynamic has a feasible starting point.

\subsection{Soundness of the \dynamicTwoCap}
\label{sec:2-soundness}

Recall that we state and prove Lemma~\ref{lem:xdot-in-kernel} and 
Theorem~\ref{thm:x-stays-feasible} general enough to include different 
choices of the conductance $G$, in particular, our new choice 
$G=X\otimes X$. Therefore, to establish soundness of our \dynamicTwo, we only need to show that $X(t)$ stays in 
the positive definite cone. In light of Lemma~\ref{lem:pd-work-function-diff},
we would like to give a lower bound on the derivative of the work function.
Unfortunately, the lower bound that we can get in the following might not 
be time-independent, at least not at first glance.

\begin{lemma}
    \label{lem:2-dynamic-work-function-bound}
    Given $X(t)\succ 0$ at some point $t$, the derivative of the work function $W(t)=\ln \det (X(t))$
    is lower bounded:
    \begin{equation}
        \label{eq:2-dynamic-workfunction-bound}
        \frac{d}{dt}\ln \det X(t) \geq -\sqrt{n}tr(CX(t))
    \end{equation}
\end{lemma}

\begin{proof}
    By definiton of $\dot{X}(t)$, we have:
    \begin{align*}
        \frac{d}{dt}\ln \det X(t) &= tr(X^{-1} (t)\dot{X}(t))\\
            & = -\vectorize{X^{-1} }(I-G\capitalize{A}^T(\capitalize{A}G\capitalize{A}^T)^{-1} \capitalize{A})
            G\vectorize{C}
    \end{align*} 
    Write $P_G\coloneqq G^{\frac{1}{2}}\capitalize{A}^T(\capitalize{A}G\capitalize{A}^T)^{-1} 
    \capitalize{A}G^{\frac{1}{2}}$ and notice that $P_G$ is an orthogonal projection. Plug it 
    into the equation above:
    \begin{align*}
        \frac{d}{dt}\ln \det X(t) &=  -\vectorize{X^{-1} }G^{\frac{1}{2}}(I-P_G)G^{-\frac{1}{2}}\vectorize{C}\\
            &\geq -\sqrt{\vectorize{X^{-1} }^TG\vectorize{X^{-1} }}\norm{(I-P_G)G^{\frac{1}{2}}\vectorize{C}}
    \end{align*}
    where the inequality follows from the Cauchy-Schwartz inequality. Now by the definition that $G=X\otimes X$,
    we get that $\vectorize{X^{-1} }^TG\vectorize{X^{-1} }=tr(I)=n$. On the other hand, since $P_G$ is an 
    orthogonal projection, it's easy to verify that $I-P_G$ is an orthogonal projection as well,
    therefore,
    \[
		\norm{(I-P_G)G^{\frac{1}{2}}\vectorize{C}} \leq \norm{G^{\frac{1}{2}}\vectorize{C}}
	\]
    Moreover,
    $\norm{G^{\frac{1}{2}}\vectorize{C}}^2=tr(XCXC)$. Now by the cyclic property of trace, we have:
    \begin{align*}
        tr(XCXC) & = tr\left ( \left (X^{\frac{1}{2}}CX^{\frac{1}{2}}\right ) 
        \left (X^{\frac{1}{2}}CX^{\frac{1}{2}}\right ) \right ) \\
        & = \sum_{i=1}^n \sigma_i^2, \text{ where } \sigma_i \text{ are the 
        eigenvalues of } \left (X^{\frac{1}{2}}CX^{\frac{1}{2}}\right )\\
        & \leq \left (\sum_{i=1}^n\sigma_i \right )^2\\
        & = tr^2\left (X^{\frac{1}{2}}CX^{\frac{1}{2}}\right )
    \end{align*}
    where in the first inequality we need the assumption that $C\succ 0$ and $X\succ 0$
    which implies that all $\sigma_i$'s are nonnegative.

    Putting all these together,
    \begin{align*}
        \frac{d}{dt}\ln\det X &= tr(X^{-1} \dot{X}) \\
            &\geq - \sqrt{n}\cdot tr\left (X^{\frac{1}{2}}CX^{\frac{1}{2}}\right )\\
            &= -\sqrt{n}\cdot tr(CX)
    \end{align*}
\end{proof}

At first glance, the lower bound that we have in Lemma~\ref{lem:2-dynamic-work-function-bound}
is time dependent, but as it turns out, it is already enough for us to prove that 
$X(t)$ stays within the positive definite cone.

\begin{theorem}
    \label{thm:2-dynamic-pd-cone}
    If $X(0)\succ 0$ is linearly feasible, then for the \dynamicTwo{} 
    where the conductance is $G=X\otimes X$, $X(t)$ stays in the positive definite cone 
    at any finite time $t$.
\end{theorem}

\begin{proof}
    If at any time $X(t)\not\succ 0$, then by the assumption that $X(0)\succ 0$ and the 
    continuity of the eigenvalues of $X(t)$, there must exists a finite time 
    when $X(t)$ obtains a zero eigenvalue. Consider the first time $T> 0$ where 
    $X(T)$ obtains a zero eigenvalue. In other words, for any $t< T, X(t)\succ 0$.
    Then by Theorem~\ref{thm:dynamic-reduce-objective}, for any $t\in [0,T),
    \frac{d}{dt}\mathcal{L}(t)\leq 0$, that is,
    \[
        \forall t\in [0,T), tr(CX(t))\leq tr(CX(0))
    \]
    Therefore, similar to the proof of Lemma~\ref{lem:pd-work-function-diff}, 
    and by Lemma~\ref{lem:2-dynamic-work-function-bound}, we have 
    \begin{align*}
        \lim_{t\to T^-}W(t) &= W(0)+\lim_{t\to T^-}\int_{s=0}^t \frac{d}{ds} W(s) ds \\
            & \geq W(0) -\sqrt{n} \cdot tr(CX(0))T,
    \end{align*}
    which is a time-independet constant, given a finite starting point $X(0)$.
    However, by our assumption that $X(T)$ obtains a zero eigenvalue, $\lim_{t\to T^-}W(t)$
    must be divergent and not lower bounded by any finite value. This yields a contradiction
    and we get that for any finite time $T\geq 0, X(T)\succ 0$.
\end{proof}

\subsection{Convergence of the \dynamicTwoCap} \label{sec:2-convergence}
Having established the soundness of the \dynamicTwo, we can 
directly conclude the following as a corollary of Theorem~\ref{thm:dynamic-reduce-objective}

\begin{corollary}
    \label{cor:2-dynamic-convergence}
    Given a feasible starting point $X(0)\succ 0$, then the \dynamicTwo{} satisfies the following:
    \begin{enumerate}
        \item $\frac{d}{dt}\mathcal{L}(t)\leq 0$
        \item $\frac{d}{dt}\mathcal{L}(t)$ becomes zero if and only if $\dot{X}(t)=0$
        \item The dynamic converges to equilibrium $\lim_{t\to \infty}\norm{\dot{X}(t)}=0$
    \end{enumerate}
\end{corollary}

Similar to before, this corollary does not immediately imply the dynamic's convergence to 
optimality. However, we show that our \dynamicTwo{} indeed converge to 
the optimum of SDP~\eqref{eq:SDP-definition} using an approach very different from the 
above. In fact, we prove that given any feasible starting point $X(0)$, the \dynamicTwo{} coincides with a central path of the SDP. This argument is akin 
to the central path argument in~\cite{straszak2015natural}, where the barrier in their 
work is an entropy type barrier. However, we use a more canonical logarithm type 
barrier. Before we proceed to the statement and proofs, we point out that the main 
reason that such arguments cannot be applied for the \dynamicOne{} 
again boils down to the noncommutative nature of matrix multiplication. Differentiating
entropy-type barriers on matrices typically results in infinite series, and analyzing them
in closed-form would require communtativity assumption. However, with a clever choice of 
conductance matrix, we manage to bypass this issue while still following our general 
recipe of \dynamicSDPs~\eqref{eq:physarum-dynamics-sdp-general}. We believe that 
it showcases the power and potential of our general framework and motivates research into 
designing other conductance matrices that leads to Physarum dynamics with better properties.

\paragraph{Central path and \dynamicSDP} Now we prove that the \dynamicTwo{} starting from a feasible point $F\succ 0$ conincides with 
a central path of the SDP. Our barrier function is 
\[
	f(X)\coloneqq -\ln \det X + tr(F^{-1}X),
\]
and with the parameter $t\geq 0$, points on our central path are 

\begin{equation}
	\label{eq:central-path}
	\begin{array}{rcc}
		X(t) \coloneqq & \text{argmin} & t\cdot tr(CX) -\ln \det X + tr(F^{-1}X)\\
			& \text{s.t.} & \capitalize{A}\vectorize{X} = b\\
			& & X \succ 0
	\end{array}
\end{equation}
As $t\to\infty$, the central path~\eqref{eq:central-path} approaches to the optimal 
solution of SDP~\eqref{eq:SDP-definition}. Now we can state the main result of this 
section:

\begin{theorem}
	\label{thm:2-dynamic-central-path}
	Fix any feasible point $F\succ 0$, the solution $X(t)$ to 
	the central path problem~\eqref{eq:central-path} is the solution to the 
	\dynamicTwo{} with the starting point $X(0)= F$.
	Consequently, the Physarum dynamic converges to the optimum of 
	SDP~\eqref{eq:SDP-definition} when starting from any feasible 
	point $X(0)\succ 0$.
\end{theorem}

\begin{proof}
	We first show that $X(0)=F$. With $t=0$, we write down the Lagrangian equations 
	of~\eqref{eq:central-path} (with the constraint $X\succ 0$ omitted):
	\begin{align*}
		-X^{-1}+F^{-1} &= \capitalize{A}^Tp\\
		\capitalize{A}\vectorize{X} &= b
	\end{align*}
	Clearly $(F,0)$ is a solution to the above system. This implies that $F$ is the 
	optimal solution to 
	$\min\left \{ -\ln \det X + tr(F^{-1}X): \capitalize{A}\vectorize{X} = b\right \}$.
	However, by our assumption that $F\succ 0$, we know that $F$ is feasible 
	for~\eqref{eq:central-path} when $t=0$. Therefore, $X(0)=F$.

	Now we prove that 
	\[
		\vectorize{\dot{X}}(t) = - (I - (X\otimes X)\capitalize{A}^T
		(\capitalize{A}(X\otimes X)\capitalize{A}^T)^{-1} \capitalize{A})
        (X\otimes X)\vectorize{C},
	\]
	which, by the definition of our Physarum dynamic, implies that the central path is 
	indeed the Physarum dynamic. Now we introduce the dual variable $y$ for the 
	linear constraints and consider the Lagrangian:
	\[
		L(X,y) = t \cdot tr(CX) -\ln \det X + tr(F^{-1}X) + y^T(\capitalize{A}\vectorize{X}-b)
	\]
	It is apriori not so clear why it is possible to drop the positive definiteness 
	constraint from the central path problem. However, once we show that the ``relaxed'' central 
	path problem coincide with the \dynamicTwo, then due to our soundness 
	statements in section~\ref{sec:2-soundness} we see that we indeed did not lose anything by dropping 
	the positive definiteness constraint. Now, we take the partial derivative of $L$ with respect 
	to $X$:
	\[
		\frac{\partial}{\partial X}L(X,y) = tC - X^{-1} + F^{-1}+\sum_\ell y_\ell A_\ell
	\]
	Setting it to zero, we get the relation between the optimal $X$ and $y$:
	\[
		X^{-1} = tC+F^{-1}+\sum_\ell y_\ell A_\ell
	\]
	Now we differentiate with respect to $t$ on both side and get 
	\[
		-X^{-1}\dot{X}X^{-1} = C +\sum_\ell \dot{y}_\ell A_\ell
	\]
	In the vectorized notation, this is 
	\[
		\vectorize{\dot{X}} = (X\otimes X)\vectorize{C}+(X\otimes X)\capitalize{A}^T\dot{y}
	\]
	However, since $\capitalize{A}\vectorize{X}=b$, we get $\capitalize{A}\vectorize{\dot{X}}=0$.
	In terms of $\dot{y}$, this becomes
	\[
		\capitalize{A}(X\otimes X)\vectorize{C}+\capitalize{A}(X\otimes X)\capitalize{A}^T\dot{y}
		= 0
	\]
	Solving for $\dot{y}$ and plugging it back, we finally obtain
	\[
		\vectorize{\dot{X}}(t) = - (I - (X\otimes X)\capitalize{A}^T
		(\capitalize{A}(X\otimes X)\capitalize{A}^T)^{-1} \capitalize{A})
        (X\otimes X)\vectorize{C},
	\]
	in other words, the central path~\eqref{eq:central-path} coincides with the 
	\dynamicTwo. 
\end{proof}

\section{Algorithms and Experimental Results}
\label{sec:algo-experiment}
In this section we introduce the augmented SDP and show that $C^{-1}$
is feasible therein, explain our discretizations of the \dynamicOne{} and hence obtain numerical algorithms for 
solving SDP, and we present empirical evaluations of the numerical algorithms 
to complement where our current theoretical techniques are lacking. Beyond that,
our empirical evaluations show the potential of nature-inspired algorithms in 
solving positive SDP. Our implementations and experiments are available \hypertarget{anchor:solverLink}{at}:

\begin{mdframed}[backgroundcolor=lightgray]
   \centering
    \solverLink
\end{mdframed}

To accurately implement an algorithm for our Physarum solver framework, we define a set of update steps whereby we calculate $\dot{X}(t)$ according to the $X(t)$ value. 
Note that according to  \eqref{eq:physarum-dynamics-sdp-general}, $\dot{X}$ can be written as,
$$\vectorize{\dot{X}} = G \capitalize{A}^T (\underset{L}{\underbrace{\capitalize{A} G \capitalize{A}^T}})^\dagger \underset{b}{\underbrace{\capitalize{A} G \vectorize{C}}} - GC$$
Recall that according to Theorem \ref{thm:xdot-update-feasible}'s proof $p$ is defined as the Lagrangian multipliers and the solution to the equation $Lp = b$. In both \dynamicOne{} and \dynamicTwo{}, $p$ can be considered as a dual candidate solution, and when the dynamic converges, it can produce a dual solution; therefore, we keep track of it to obtain a primal-dual solver.

That said, we split up the computation of the update step into three phases which repeat in all extensions of our SDP solvers:
\begin{enumerate}
    \item We calculate matrix $L$ element-by-element. $L_{i,j}$ is equal to $\vectorize{A_i}^T G \vectorize{A_j}$ in the general framework. For the \dynamicOne, it is equal to $tr(C^{-1} A_i X A_j)$ and for the \dynamicTwo{} it is equal to $tr(X A_i X A_j)$.
    \item We then calculate the solution to $Lp = b$ and present $p$ as a dual candidate solution. In general, $p = L^\dagger b$.
    \item Finally, we put it all together by calculating $\dot{X}$ using the following formula:
    $$\vectorize{\dot{X}} = G \capitalize{A}^T p - GC.$$
    With the \dynamicOne, the formula can be re-written in the following matrix notation without the need for vectorization:
    \begin{equation} \label{eq:1st-ansatz-simplified-update}
        \dot{X} = \sum_{\ell = 1}^m p_\ell \frac{C^{-1}A_\ell X + X A_\ell C^{-1}}{2} - X
    \end{equation}
    By analogy, one can also write down the following simplified formulation for the \dynamicTwo{} in a similar fashion:
    \begin{equation} \label{eq:2nd-ansatz-simplified-update}
    \dot{X} = \sum_{\ell=1}^m p_\ell X A_\ell X - X C X
    \end{equation}
    We omit the steps of computations here.

\end{enumerate}

\subsection{Augmentation}
\label{sec:augmentation}
Recall that in section~\ref{sec:1-dynamic}, our theoretical guarantees relies on 
the assumption that $C^{-1}$ is linearly feasible. While we conjecture that such assumption 
is unnecessary, in this section we show that SDP~\eqref{eq:SDP-definition} can be augmented 
so that the inverse of the objective matrix is linearly feasible. To this end, we augment 
all matrices in~\eqref{eq:SDP-definition} by one row and column. We define $\bar{C}$ and $\bar{A}_\ell$ 
according to the following scheme using a certain $\gamma$:
\[
	\bar{C} = \begin{pmatrix} \gamma C & 0 \\ 0 & 1 \end{pmatrix}, \bar{A}_i = \begin{pmatrix} A_i & 0 \\ 
		0 & \alpha_i \end{pmatrix}~~,~~\alpha_i = b_i - \frac{tr(A_i C^{-1})}{\gamma}.
\]
We will also refer to an arbitrary matrix $T$ as \textit{augmented} if all entries of the 
right column and bottom row except bottom right entry are zero. 
With these, define the following SDP the \textbf{Augmented SDP}:
\begin{equation} 
	\label{eq:augmented-SDP-defenition}
	\min \{ tr( \bar{C} \bar{X} ) : tr( \bar{A}_{\ell} \bar{X} ) = 
		b_\ell~\forall \ell \in [m],\, \bar{X} \succeq 0 \}.
\end{equation}

One can easily show $\bar{C}^{-1}$ is feasible in this case:
$$tr(\bar{A}_\ell \bar{C}^{-1}) = \frac{A_\ell C^{-1}}{\gamma} + \alpha_\ell = b_\ell$$
Therefore, by choosing $\bar{X}(0) = \bar{C}^{-1}$ we start with a feasible solution.
By Theorem~\ref{thm:x-stays-feasible}, $X(t)$ remains linearly feasible. Next, we present a lemma that states if $X(0)$ is augmented, then $X(t)$ stays augmented under certain conditions. We omit the proof here because it boils down to tedious verification and offers no insights. This lemma establishes a mapping 
between the dynamic in the augmented problem and the original problem.
\begin{lemma}
	\label{lem:1-dynamic-x-stays-augmented}
	For any given autonomous dynamic on $\bar{X}(t)$ with a symmetric and augmented starting point $\bar{X}(0)$ where
	$$\dot{\bar{X}}(t) = v(\bar{X}(t)),$$
	$X(t)$ remains augmented throughout the dynamic if $v(X)$ is obtained by addition, multiplication, and inversion over symmetric augmented matrices including $\bar{X}$.
\end{lemma}

By taking into account the simplified formulas \eqref{eq:1st-ansatz-simplified-update} and \eqref{eq:2nd-ansatz-simplified-update}, both of these dynamic comply with the dynamics described in Lemma \ref{lem:1-dynamic-x-stays-augmented} in the augmented case. Therefore, we can re-write the dynamic in the augmented form as below:
\[
	\bar{X}(t) = \begin{pmatrix}
	\Tilde{X}(t) & 0\\
	0 & \beta(t)
	\end{pmatrix}, \bar{\dot{X}}(t) =
	\begin{pmatrix}
	\dot{\Tilde{X}}(t) & 0\\
	0 & b^T p -  \beta(t)
	\end{pmatrix}
\]
Each feasible solution of the original problem maps to a feasible solution $\bar{X}$ of the augmented problem 
by adding a row and a column with all zeros. The objective value of such an $\bar{X}$ is then scaled by $\gamma$, 
i.e., $tr( \bar{C} \bar{X} ) = \gamma tr( C X )$. By an appropriate choice of $\gamma$, these objective values can 
be made smaller than the objective value of $\bar{C}^{-1}$, which is a constant. Hence we have the following:
\begin{equation} \label{eq:side_1_augmented}
tr(\bar{C}\bar{X}_{opt}) \le \gamma tr(C X_{opt})
\end{equation}

If the dynamic converges to the optimum of the 
augmented SDP, and $\beta(t)$ converges to zero, then at sufficiently large time $T>0$,
\begin{align*}
	\gamma tr(CX_{opt}) &\geq tr(\bar{C}\bar{X}_{opt})\\
		& = tr(\bar{C}\bar{X}) -\epsilon\gamma\\
		& = \gamma tr(\Tilde{X}) +\beta(T) - \epsilon\gamma\\
		&\geq \gamma tr(\Tilde{X})- \epsilon\gamma
\end{align*}
Therefore, the approximation $\Tilde{X}$ has objective at most $\epsilon$
larger than the optimum objective value. Note that $\Tilde{X}$ is almost feasible
where the error is proportional to $\beta(t)$.
So with the augmentation, we are approximating the original problem, 
both in terms of optimality and feasibility. We will also provide experimental evidence in Section \ref{sec:experiment} that shows $\beta(t)$ approaches zero for \dynamicOne{} where we need augmentation. Note that $\bar{X}(0) = \bar{C}^{-1}$ complies with the condition in Conjecture \ref{conj:general-dynamic-convergence} as it is equal (therefore dominates) a feasible solution: $\bar{C}^{-1}$. Furthermore, this also provides evidence to back up the optimality convergence statement in Conjecture \ref{conj:general-dynamic-convergence}.

\subsection{Discretization and Numerical Algorithms} 

In this section, we investigate the discretization of the \dynamicSDP{} and introduce a framework to obtain a primal dual solver for SDP problems. The algorithms and experimental evaluations in this section complement our conjectures about the soundness and convergence to optimality of the \dynamicOne.
Moreover, we believe that these nature-inspired numerical algorithms have strong potentials in its practicality.

To simulate the continuous dynamic in an algorithm, we will discretize it using the 
following~\cite{butcher2016numerical}:
\begin{equation}
	\label{eq:discretization}
	X(t+1) \gets X(t) + h \dot{X}(t)
\end{equation}
Intuitively, this equation simulates the \dynamicSDP{} when $h \to 0$. As a consequence, for practical purpose, $h$ should be chosen small enough.

\paragraph{Vanilla algorithm} The vanilla algorithm is a straightforwad implementation of 
the discretization scheme~\eqref{eq:discretization}. The algorithm consists of a sequence of 
iterations where in each iteration we solve for $\dot{X}(t)$ and update $X(t)$ with the discretized 
dynamic and step $h$ to obtain $X(t+1)$. The process continues until $\dot{X}(t)$ becomes too
small which indicates reaching an equilibrium point. We stop the iteration whenever 
$\norm{\dot{X}}$ is smaller than some small constant $\epsilon$. Thus we obtain the following algorithm:

\begin{algorithm}[H]
	\caption{Vanilla Physarum SDP solver}	
	\label{alg:vanilla-SDP}
	\SetKwInOut{Input}{Input}\SetKwInOut{Output}{Output}
	\Input{$C, {A}_1,\ldots,{A}_m \in S_n$, $b \in \mathbb{R}^m, X(0)$.}
	\Output{$(X^{eq} \succeq 0, p^{eq})$.}
	\BlankLine
	Let $t =0$\;
	\Repeat{$\norm{\dot{X}(t)} \leq \epsilon$}{
	$\dot{X}(t), p(t) \gets \textsc{SolveUpdateProblem}(C, A_1, ..., A_m, b, X(t))$ using 
		Algorithm~\ref{alg:Solve-Update-Problem}\;
	Calculate small enough $h$\;
	Update $X(t+1) \gets X(t) + h \dot{X}(t)$\;
	Increment $t$\;
	}
	\Return $(X(t), p(t))$
\end{algorithm}

Each iteration solves an update problem (or equivalently, a projection under weighted 
inner product). We give below the update problem solver in terms of $\dot{X}$ and $p$; moreover, we use the simplified formulas in \eqref{eq:1st-ansatz-simplified-update} to implement the update problem.

\begin{algorithm}[H]
	\caption{Solve Update Problem}	
	\label{alg:Solve-Update-Problem}
	\SetKwInOut{Input}{Input}\SetKwInOut{Output}{Output}
	\Input{$C, {A}_1,\ldots,{A}_m \in S_n$, $b \in \mathbb{R}^m, X$.}
	\Output{$\dot{X}, p$.}
	\BlankLine
	Calculate the $m \times m$ matrix $L$ and let $L_{i,j} \gets \vectorize{A_i}^TG\vectorize{A_j}$\;
	Calculate $p \gets L^\dagger b$\;
	Calculate $\dot{X}$ s.t. $\vectorize{\dot{X}} \gets  G \capitalize{A}^T p - GC$\;
	\Return $(\dot{X}, p)$
\end{algorithm}

As mentioned before, $h$ should be chosen small enough. One obvious upper-bound for $h$ 
is that the update should preserve positive semi-definiteness of $X(t)$. There is 
a close form bound that guarantees positive semi-definiteness, or one can use exponential/binary
search. We defer interested reader 
to our \hyperlink{anchor:solverLink}{implementation}.

In our experiments, we run Algorithm~\ref{alg:vanilla-SDP} on the augmented problems 
and map the augmented solution back to the original problem (recall that by 
Lemma~\ref{lem:1-dynamic-x-stays-augmented} the iterates stays augmented). We provide empirical evidence that $\beta(t)$ converges to zero and that $\Tilde{X}$ which was obtained from taking the upper-left $n \times n$ matrix of $X$ yields a sound and optimal solution to our SDP.

\paragraph{Modified algorithm} One of the problems encountered by Algorithm~\ref{alg:vanilla-SDP}
is that $L=\capitalize{A}G\capitalize{A}^T$ becomes ill-conditioned as some eigenvalues of $X$, and in turn, $G$ converge to zero while reaching an equilibrium point which is not full-rank. To circumvent this problem, we introduce the modified algorithm that facilitate numerical hacks to work with near-zero eigenvalues.

For the modified algorithm, we drop the condition on $X(0)$ being linearly feasible and we start with a large enough matrix $X(0) = \eta \times I$ where $\eta$ is set to a large value such that $X(0)$ dominates an arbitrary feasible solution; this complies with the statement in Conjecture \ref{conj:general-dynamic-convergence} where we consider non-feasible starting points. As a result, we do not augment the problem, and interestingly enough, the dynamic still converges to the optimum in our experiments. We introduce a set of ``epochs'' in our algorithm. Each epoch ends whenever either $\norm{\dot{X}(t)}$  becomes less than a certain $\epsilon$ (i.e, the dynamic converges) or the minimum eigenvalue of $X(t)$ becomes smaller than $\epsilon$. In the former case, we have reached equilibrium, but in the latter, we need to restart the dynamic and start another epoch.

We keep track of a set of linearly independent vectors, which we call a ``diagonalization basis'', that can simultaneously diagonalize $X(t)$ and $C$ at any time $t$. The diagonalization basis will contain $n$ vectors in the beginning since the initialization sees $X(0)$ as being full-rank. Assume that at a certain arbitrary epoch, $X(0)$ has rank $k$. Let us identify the diagonalization basis with $\{\Tilde{u}_1, \Tilde{u}_2, ..., \Tilde{u}_k\}$ and obtain the $n \times k$ matrix $\Tilde{U}$ by stacking these vectors in columns next to each other. $\Tilde{U}$ is a time-invariant matrix that diagonalizes $C$ and $X(t)$ at any time $t$; in other words, both $\Tilde{U}^T C \Tilde{U}$ and $\Tilde{U} X(t) \Tilde{U}$ are diagonal. The detail on why this matrix $\Tilde{U}$ exists and how to obtain it is spared and we defer readers to take a look into our \hyperlink{anchor:solverLink}{implementation} which calculates these matrices using a simultaneous diagonalization technique on $C$ and $X(0)$ based on generalized eigenvalues~\cite{ghojogh2022eigenvalue} at the beginning of each epoch.

With that in mind, when $X(t)$ obtains a near zero eigenvalue, the diagonal matrix $\Tilde{U}^T X(t) \Tilde{U}$ will obtain a near zero element $\lambda_i$ on the diagonal. In that case, we pop $\Tilde{u}_i$ from the diagonalization basis and start the new epoch.
We define a projection of an arbitrary matrix $M$ on the set of diagonalization basis vectors $\{\Tilde{u}_1, ..., \Tilde{u}_k\}$ as the solution to the following equation:

$$projection_{\Tilde{U}}(M) = \sum_{i = 1}^k c_i \Tilde{u}_i \Tilde{u}_i^T \text{ s.t } \{c_1, ..., c_m\} =  \underset{ \{z_1, ..., z_m\}}{\arg\min} \norm{M - \sum_{i=1}^k z_i \Tilde{u}_i\Tilde{u}_i^T}_2^2$$

Intuitively speaking, the $projection$ function considers the best approximation of a matrix given a certain set of vectors. When $\Tilde{u}_i^T X(t) \Tilde{u}_i$ converges to zero, we may assume that $\Tilde{u}_i$ is in the kernel of any optimum $X$; therefore, projecting all of the dynamic into a new basis by excluding $\Tilde{u}_i$ will help avoid dealing with near zero eigenvalues in $X(t)$. As a result, we can work around numerical instability by calculating a projected $G$ matrix in the update problem phase which will not be ill-conditioned. 

The modified algorithm contains a lot of technical details embedded in our implementation; however, for simplicity, we do not expand on the theoretical foundations of these technical details. A simplified version of the algorithm using the projection intuition can be summarized in Algorithm~\ref{alg:modified-SDP}. 
Additionally, we use the modified algorithm as strong evidence to back up Conjecture \ref{conj:general-dynamic-convergence}. Note that even though we do not start from a feasible solution, we converge to an optimum solution when we set $X(0)$ to be a large enough matrix.

\begin{algorithm}
	\caption{Modified Physarum SDP solver}	
	\label{alg:modified-SDP}
	\SetKwInOut{Input}{Input}\SetKwInOut{Output}{Output}
	\Input{$C, {A}_1,\ldots,{A}_m \in S_n$, $b \in \mathbb{R}^m$, $\eta$.}
	\Output{$(X^{eq} \succeq 0, p^{eq})$.}
	\BlankLine
	Set $X(0) = \eta \times I$\;
	Initialize the diagonalization basis in the $n\times n$ matrix $\Tilde{U}$ s.t. $\Tilde{U}^T X(t) \Tilde{U}$ and $\Tilde{U}^T C \Tilde{U}$ remain diagonal throughout the dynamic\;
	
	\Repeat{$\norm{\dot{X}(t)} \le \epsilon$ and $\lambda_{\min}(projection(X(t))) \ge \epsilon$}{
	Pop any column in $\Tilde{U}$ which corresponds to diagonal elements less than $\epsilon$ in $X(0)$\;
	Set $X(0) \gets \eta \times projection_{\Tilde{U}}(X(0))$\;
	Set $C \gets projection_{\Tilde{U}}(C)$\;
	Let $t \gets0$\;
	\Repeat{$\norm{\dot{X}(t)} \le \epsilon$ or $\lambda_{\min}(X(t)) \le \epsilon$}{
	Let $\dot{X}(t), p(t) \gets \textsc{SolveUpdateProblem}(C, A_1, ..., A_m, b, X(t))$ using Algorithm \ref{alg:Solve-Update-Problem}\;
	Calculate small enough $h$\;
	Update $X(t+1) \gets X(t) + h \dot{X}(t)$\;
	Increment $t$\;
	}
	}
	\Return $(X(t), p(t))$
\end{algorithm}

\subsection{Experimental Evaluation}
\label{sec:experiment}
In this section we discuss the dataset used to evaluate our algorithms and 
their implementations. These experimental results supports our Conjecture~\ref{conj:general-dynamic-convergence}.
For validation, our algorithms are compared with a standard 
SDPASolver\footnote{\url{http://sdpa.sourceforge.net/}} as ground truth.

We evaluate the objective value obtained from our dynamic using a metric identified as ``gap" which calculates the absolute difference between the ground truth primal solution $tr(CX^*)$ and the objective value obtained from our dynamic: $tr(CX(\infty))$. We also evaluate the soundness of our algorithm according to a real number identified as ``infeasibility". The infeasibility metric is calculated as below:
$$infeasibility(X) = \max \left( \underset{\ell \in [m]}{\max} |b_\ell - tr(A_\ell X)|, \max(0, -\lambda_{\min}(X)) \right)$$

For a sound dynamic, infeasibility should converge to zero and any deviation from zero impedes the soundness of the dynamic.

\paragraph{The dataset} Table \ref{table:dataset} shows a set of $180$ positive SDP samples. 
The tests are generated using three different schemes on different matrix size initialization. 
Details as follows:

\begin{itemize}
    \item[-] \textbf{Random-Tests:} These tests are generated randomly. 
		Given a certain $n$ and $m$ - the size of matrices and the number of 
		conditions - a random positive definite matrix $C$ and random symmetric 
		matrices $A_\ell$ are generated. The datasets which are created using this scheme 
		consist of ``testset1" to ``testset3", ``large1", and ``large2".
    \item[-] \textbf{Vertex-Cover:} To generate these data we create a random graph with $n$ 
		vertices numbered from $1$ to $n$ and $m$ edges. We will then create an SDP which is able to 
		approximate the minimum vertex cover problem. For each random graph with $n$ vertices, $C$ is the identity matrix of size $(n + 1) \times (n+1)$. 
		Then the following linear condition matrices are generated:
		\begin{itemize}
			\item For each vertex $v$ in the graph an $A_\ell$ is created according to the following:
			$$(A_\ell)_{i,j} = \begin{cases}
			-1 & \text{if } \min(i,j) = 1 \text{ and } \max(i,j) = v + 1\\
			2 & \text{if } i = j = v + 1\\
			0 & \text{otherwise}
			\end{cases}$$
			and the respective $b_\ell$ is set to $0$.
			\item For each edge $(v, w)$ in the graph an additional $A_\ell$ is created according to the following:
			$$(A_\ell)_{i,j} = \begin{cases}
			-1 & \text{if } \min(i,j) = \min(v + 1, w + 1) \text{ and } \max(i, j) = \max(v + 1, w + 1)\\
			1 & \text{if } \min(i,j) = 1 \text{ and } \max(i,j) \in \{v + 1, w + 1\}\\
			0 & \text{otherwise}
			\end{cases}$$
			and the respective $b_\ell$ is set to $2$.
			\item An additional $A_\ell$ is created which is set to zero except one element 
				on the upper-left corner and the respective $b_\ell$ is $1$.
		\end{itemize}
    These tests have $n + m + 1$ conditions and matrices of size $n + 1$.

    \item[-] \textbf{Max-Cut:} We have created two sets of Max-Cut problems referred to as `maxcut1' and `maxcut2'. 
	To generate them, we create a random graph with $n$ vertices. An approximation of the maximum cut of the created 
	graph can be computed using the following SDP problem 
    $$\arg\max\{tr(W X) ~:~ X \succeq 0, \text{Main diagonal entries of } X \text{ are all 1}\}, $$
    where $W$ is the Laplacian matrix of the graph. To satisfy the constraints of our algorithm, 
	we transform the aforementioned to
    $$\arg\min\{ tr\left(\left[\xi I - W\right] X\right)~:~ X \succeq 0, 
		\text{Main diagonal entries of } X \text{ are all 1}\},$$
    where $\xi$ is big enough so that $\left[\xi I - W \right]\succ 0$.
	Note that as all main diagonal entries of $X$ are required to be $1$ which correspond to $A_\ell$ 
	matrices that are set to be entirely zero except the entry on $(i,i)$. 
	That said, these tests have $n$ linear constraints and matrices of size $n$.
\end{itemize}

\begin{table}[ht]
	\centering
	\begin{tabular}{|c|c|c|}
	\hline
	TestSet Description                                                                                                           & TestSet Name   & No. of tests \\ \hline
	\begin{tabular}[c]{@{}c@{}}Randomly generated definite SDP samples\\ with $n=5, m \in [3] $\end{tabular}                     & testset1       & 20           \\ \hline
	\multirow{2}{*}{\begin{tabular}[c]{@{}c@{}}Randomly generated definite SDP samples\\ with $n=10, m \in [5]$\end{tabular}}   & testset2-1     & 14           \\ \cline{2-3} 
																																  & testset2-2     & 16           \\ \hline
	\multirow{2}{*}{\begin{tabular}[c]{@{}c@{}}Randomly generated definite SDP samples\\ with $n=25, m \in [10]$\end{tabular}} & testset3-1     & 16           \\ \cline{2-3} 
																																  & testset3-2     & 4           \\ \hline
	\multirow{2}{*}{\begin{tabular}[c]{@{}c@{}}Randomly generated vertex cover problems\\ with $|V(G)| = 5, |E(G)|\in [10]$\end{tabular}}       & vertexcover1-1 & 13             \\ \cline{2-3} 
																																  & vertexcover1-2 &    7          \\ \hline
	\multirow{2}{*}{\begin{tabular}[c]{@{}c@{}}Randomly generated vertex cover problems\\ with $|V(G)| = 20, |E(G)|\in [130]$\end{tabular}}        & vertexcover2-1 &   4           \\ \cline{2-3} 
																																  & vertexcover2-2 &    16          \\ \hline
	\begin{tabular}[c]{@{}c@{}}Randomly generated definite SDP samples\\ with $n=50, m \in [5, 10]$\end{tabular}                  & large1         & 10           \\ \hline
	\begin{tabular}[c]{@{}c@{}}Randomly generated definite SDP samples\\ with $n=100, m \in [5, 20]$\end{tabular}                 & large2         & 20           \\ \hline
	\begin{tabular}[c]{@{}c@{}}Randomly generated vertex cover problems\\ with $|V(G)| = 50, |E(G)| \in [10, 20]$\end{tabular}    & vertexcover3   & 10           \\ \hline
	\begin{tabular}[c]{@{}c@{}}Randomly generated max cut problems\\ with $|V(G)| \in [20, 50]$ \end{tabular}    & maxcut1   & 20           \\ \hline
	\begin{tabular}[c]{@{}c@{}}Randomly generated vertex cover problems\\ with $|V(G)| = 100$\end{tabular}    & maxcut2   & 10           \\ \hline
	\end{tabular}
	\caption{Standard SDPLib dataset that is generated to validate the Physarum dynamics.}
	\label{table:dataset}
\end{table}
	
\paragraph{Experimental results} Now we explain our empirical results in details,
thus providing evidence that \dynamicOne{} is sound and converges to 
optimum, in the augmented case and unconditionally (see Conjecture~\ref{conj:general-dynamic-convergence}).
We discuss the performances of the Algorithm~\ref{alg:vanilla-SDP} for the augmented case and Algorithm~\ref{alg:modified-SDP} for the general case. We extract the maximum $gap$ and $infeasibility$ for each of our test sets and will also provide empirical evidence that $\beta(t)$ converges to zero in the augmented case for Algorithm~\ref{alg:vanilla-SDP}.

\paragraph{Performance of the vanilla algorithm} 
Table~\ref{tab:vanilla-augmented} contains the results of Algorithm~\ref{alg:vanilla-SDP} on each of the datasets.
Each row contains one of the datasets in Table~\ref{table:dataset}. The solver 
provides accurate solutions on most of the datasets. In the ``vertexcover1" and ``vertexcover2" dataset 
however we have an increase in the number of linear conditions which makes the matrix $L$ larger; therefore, making the update step more difficult. 
This can in turn make the algorithm both slower and more numerically unstable. 
In addition to that, the solver runs into numerical difficulties in these tests. 
We believe that there are two main causes:
\begin{itemize}
    \item $h$ is not chosen sufficiently small. We have only used a naive approach to choose $h$ and by 
		tuning it we might get more accurate results.
    \item $L$ becomes ill-conditioned due to some of $\lambda_i(t)$ values vanishing. 
\end{itemize}
The latter is addressed by the modified algorithm~\ref{alg:modified-SDP}.

The vanilla algorithm starts with a linearly feasible $X(0)$ by augmenting the problem. We have presented the maximum $\beta$ values in the last iteration in the table. It is worthwhile mentioning that we did not augment 
the problem in ``maxcut1" and ``maxcut2" 
datasets as the identity matrix is a feasible solution for these tests; therefore, by setting $X(0) = I$ we obtain the linear feasibility condition which was the main motive for our augmentation technique. On the other hand, for the rest of the dataset, 
$\beta$ values converged to zero closely, indicating that we manage to 
approach linear feasibility in the original problem.

\begin{table}[ht] 
	\begin{footnotesize}
	\begin{tabular}{|c|c|c|c|c|c|}
	\hline
	TestSet        & \begin{tabular}[c]{@{}c@{}}Ratio of tests\\ with error \\ below $10^{-2}$\end{tabular} & \begin{tabular}[c]{@{}c@{}}Average Time\\  Spent\end{tabular} & \begin{tabular}[c]{@{}c@{}}Maximum primal gap\\ on tests with a lower \\ than $10^{-2}$ primal gap\end{tabular} & \begin{tabular}[c]{@{}c@{}}Maximum \\ Infeasibility\end{tabular} & \begin{tabular}[c]{@{}c@{}}Maximum $\beta$\\ in last iteration for \\ high accuracy tests\end{tabular} \\ \hline
	testset1       & 20/20                                                                                 & 0.635 (sec)                                                   & $2.7 \times 10^{-8}$                                                                                           & $7.7 \times 10^{-12}$                                            & $7.4 \times 10^{-54}$                                                                                  \\ \hline
	testset2-1     & 14/14                                                                                 & 0.335 (sec)                                                   & $9.5 \times 10^{-7}$                                                                                           & $7.9 \times 10^{-12}$                                            & $9.5 \times 10^{-7}$                                                                                   \\ \hline
	testset2-2     & 6/6                                                                                   & 1 (sec)                                                       & $1.3 \times 10^{-6}$                                                                                           & $1.2 \times 10^{-11}$                                            & $3.8 \times 10^{-114}$                                                                                 \\ \hline
	testset3-1     & 14/16                                                                                 & 1.26 (sec)                                                    & $7.9 \times 10^{-7}$                                                                                           & $1.8 \times 10^{-11}$                                            & $9.3 \times 10^{-30}$                                                                                  \\ \hline
	testset3-2     & 4/4                                                                                   & 4.3 (sec)                                                     & $4.7 \times 10^{-7}$                                                                                           & $1.8 \times 10^{-11}$                                            & $7.3 \times 10^{-189}$                                                                                 \\ \hline
	vertexcover1-1 & 13/13                                                                                 & 1.4 (sec)                                                     & $7 \times 10^{-4}$                                                                                             & $0$                                                              & $7.4 \times 10^{-5}$                                                                                   \\ \hline
	vertexcover1-2 & 5/7                                                                                   & 1.8 (sec)                                                     & $10^{-3}$                                                                                                      & $7.2 \times 10^{-5}$                                             & $10^{-3}$                                                                                              \\ \hline
	large1         & 10/10                                                                                 & 7.1 (sec)                                                     & $1.2 \times 10^{-8}$                                                                                           & $3.2 \times 10^{-11}$                                            & $7 \times 10^{-54}$                                                                                    \\ \hline
	large2         & 20/20                                                                                 & 25.8 (sec)                                                    & $1.4 \times 10^{-6}$                                                                                           & $6 \times 10^{-11}$                                              & $1.3 \times 10^{-6}$                                                                                   \\ \hline
	vertexcover2-1 & 1/4                                                                                   & 60 (sec)                                                      & $2.6 \times 10^{-9}$                                                                                           & $0$                                                              & $1.0 \times 10^{-7}$                                                                                   \\ \hline
	vertexcover3   & 4/13                                                                                  & 110 (sec)                                                     & $4.8 \times 10^{-5}$                                                                                           & $0$                                                              & $2.2 \times 10^{-3}$                                                                                   \\ \hline
	vertexcover2-2 & 2/16                                                                                  & 196 (sec)                                                     & $3.3 \times 10^{-9}$                                                                                           & $0$                                                              & $2.6 \times 10^{-3}$                                                                                   \\ \hline
	maxcut1 & 20/20                                                                                  & 13.4 (sec)                                                     & $1.04 \times 10^{-5}$                                                                                           & $0$                                                              & $-$                                                                                   \\ \hline
	maxcut2 & 10/10                                                                                  & 155 (sec)                                                     & $2.03 \times 10^{-5}$                                                                                           & $0$                                                              & $-$                                                                                   \\ \hline
	\end{tabular}
	\end{footnotesize}
	\caption{Results of the vanilla algorithm after augmenting the matrices with $\gamma = 0.01$. A test is considered accepted if the primal gap between the Physarum algorithm primal objective value and the ground truth is less than $10^{-2}$.}
	\label{tab:vanilla-augmented}
\end{table}

\paragraph{Performance of the modified algorithm}
We test the performance of Algorithm~\ref{alg:modified-SDP} without augmentation.
Table~\ref{table:results-modified} shows the results for the modified dynamics on the datasets. 
This solver is able to solve all of the tests in the dataset with descent accuracy, thus 
providing empirical support for our conjecture on the unconditional convergence 
of \dynamicOne.
However, it takes a longer time for the dynamic to converge accurately. 
We believe by a more appropriate choice of $h$ one can also obtain a fast and more accurate Physarum solver.

Note that although we have claimed that the accuracy of this algorithm is better than Vanilla, 
Vanilla works better on the maxcut tests which make it a really good candidate to tackle these problems. 
In addition, ``maxcut1" and ``maxcut2" datasets contain relatively large matrices which is the reason 
why in some tests in order to get the desired accuracy the algorithm took approximately $38$ minutes 
to become decently accurate.

\begin{table}[htb]
	\centering
	\begin{tabular}{|c|c|c|c|}
	\hline
	Testset Name   & \begin{tabular}[c]{@{}c@{}}Average\\ Time Spent\end{tabular} & Maximum Primal Gap    & Maximum Infeasibility \\ \hline
	testset1       & 7.5 (sec)                                                    & $3.6 \times 10^{-5}$  & $5.9 \times 10^{-12}$ \\ \hline
	testset2-1     & 2.7 (sec)                                                    & $2.9 \times 10^{-6}$  & $4.8 \times 10^{-6}$  \\ \hline
	testset2-2     & 65 (sec)                                                     & $2.48 \times 10^{-5}$ & $1.1 \times 10^{-13}$ \\ \hline
	testset3-1     & 52 (sec)                                                     & $1.9 \times 10^{-7}$  & $5.1 \times 10^{-6}$  \\ \hline
	testset3-2     & 115 (sec)                                                    & $8 \times 10^{-4}$    & $5 \times 10^{-7}$    \\ \hline
	vertexcover1-1 & 9 (sec)                                                      & $1.3 \times 10^{-6}$  & $2.4 \times 10^{-7}$  \\ \hline
	vertexcover1-2 & 41 (sec)                                                     & $5.9 \times 10^{-4}$  & $2.6 \times 10^{-4}$  \\ \hline
	vertexcover2-1 & 8 (min)                                                      & $4 \times 10^{-6}$    & $3.88 \times 10^{-6}$ \\ \hline
	vertexcover2-2 & 22 (min)                                                     & $1.8 \times 10^{-2}$  & $3.3 \times 10^{-5}$  \\ \hline
	large1         & 32 (sec)                                                     & $1.02 \times 10^{-4}$ & $5.9 \times 10^{-7}$  \\ \hline
	large2         & 170 (sec)                                                    & $7.2 \times 10^{-4}$  & $5.5 \times 10^{-11}$ \\ \hline
	vertexcover3   & 7 (min)                                                      & $1.7 \times 10^{-4}$  & $1.8 \times 10^{-4}$  \\ \hline
	maxcut1   & 3.5 (min)                                                      & $9.4 \times 10^{-3}$  & $1.2 \times 10^{-9}$  \\ \hline
	maxcut2   & 38 (min)                                                      & $9.1 \times 10^{-3}$  & $1.2 \times 10^{-9}$  \\ \hline
	\end{tabular}
	\caption{Results of the modified Algorithm \ref{alg:modified-SDP} on the dataset. Maximum Primal Gap is the difference between the Physarum SDP solver and the SDPA baseline. Infeasibility measures how much $X$ conflicts with the constraints. It is calculated by taking the maximum magnitude of the negative eigenvalue of $X^{eq}$ and the maximum $|b_\ell - tr(A_\ell X)|$ for $1 \le \ell \le m$.}
	\label{table:results-modified}
\end{table}

\clearpage

\bibliographystyle{ieeetr}
\bibliography{references}

\end{document}